\documentclass{article}
\usepackage[utf8]{inputenc}

\usepackage[margin=0.5in]{geometry}
\usepackage{graphicx}
\usepackage[usenames,dvipsnames]{xcolor}

\usepackage{amsmath,amsthm,amsfonts}
\usepackage[utf8]{inputenc}
\usepackage{dsfont}

\usepackage{subcaption}
\usepackage[hidelinks]{hyperref}
\usepackage[nameinlink]{cleveref}
\usepackage{titling}

\newcommand{\payoff}[2]{\pi_{#1,#2}}

\crefname{equation}{Eq.}{Eqs.}

\newtheorem{theorem}{Theorem}
\newtheorem{lemma}[theorem]{Lemma}
\newtheorem{claim}[theorem]{Claim}
\crefname{claim}{Claim}{Claims}
\newtheorem{observation}[theorem]{Observation}
\crefname{observation}{Observation}{Observations}

\newcommand{\pa}[1]{\left( #1 \right)}
\newcommand{\bbE}{\mathbb{E}}
\newcommand{\bbR}{\mathbb{R}}
\newcommand{\calB}{\mathcal{B}}

\newcommand{\producer}{\mathcal{P}}
\newcommand{\scrounger}{\mathcal{S}}

\begin{document}
\title{Enhanced Food Availability can Deteriorate Fitness \\ through Excessive Scrounging}
\author{Robin Vacus\thanks{CNRS, located at the Research Institute on the Foundations of Computer Science (IRIF), Paris, France  (e-mail: rvacus@irif.fr).} ~and Amos Korman\thanks{CNRS, located at the French-Israeli Laboratory on Foundations of Computer Science, UMI FILOFOCS, CNRS, UP7, TAU, HUJI, WIS International Joint Research Unit, Tel-Aviv, Israel (e-mail: amos.korman@irif.fr). }
}
\date{}
\maketitle
\begin{abstract}
In group foraging situations, the conventional expectation is that increased food availability would enhance consumption, especially when animals prioritize maximizing their food intake. This paper challenges this conventional wisdom by conducting an in-depth game-theoretic analysis of a basic producer-scrounger model, in which animals must choose between intensive food searching as producers or moderate searching while relying on group members as scroungers. Surprisingly, our study reveals that, under certain circumstances, increasing food availability can amplify the inclination to scrounge to such an extent that it paradoxically leads to a reduction in animals' food consumption compared to scenarios with limited food availability. We further illustrate a similar phenomenon in a model capturing free-riding dynamics among workers in a company. We demonstrate that, under certain reward mechanisms, enhancing workers' production capacities can inadvertently trigger a surge in free-riding behavior, leading to both diminished group productivity and reduced individual payoffs.  Our findings underscore the significance of contextual factors when comprehending and predicting the impact of resource availability on individual and collective outcomes. 
\end{abstract}

\subsection*{} \label{sec:intro}

\noindent Braess's paradox, a thought-provoking result in game theory, demonstrates that in certain transportation networks, adding a road to the network can paradoxically increase traffic latency at equilibrium \cite{braess1968paradoxon,roughgarden2005selfish,case2019braess,cohen1991paradoxical}. In a similar vein, our study demonstrates how  improvement in underlying conditions, which may initially seem beneficial, can actually lead to degraded performance. However, instead of focusing on network flows as in Braess's paradox, our findings relate to contexts of productive groups, highlighting the impact of free-riding behavior.

Productive groups, such as workers in a company, collaborating researchers, or ensembles of foraging animals, consist of individuals who not only benefit from their own resource generation or findings, but also enjoy the added advantage
of reaping the rewards of others' contributions
\cite{kim1984free,albanese1985rational,holmstrom1982moral,giraldeau2000social,giraldeau2008social}.
For example, workers in a company may receive performance-based bonuses as a reward for their productivity, while also benefiting from the collective production of their peers, through stock shares or other profit-sharing mechanisms.
Similarly, in the realm of joint research projects, the success of the endeavor contributes to the collective prestige of the researchers, yet those who make substantial contributions often receive heightened recognition and prestige.  Likewise, in group foraging scenarios, animals that first discover food patches often have an opportunity to feed before other group members join in, granting them the ability to directly consume a portion of the food they found and secure a larger share
\cite{mottley2000experimental,aplin2017stable,morand2007wild,di2001social}.

Within such productive group contexts, the pervasive occurrence of free-riding becomes apparent \cite{baumol2004welfare,hardin1968tragedy,kim1984free,albanese1985rational,giraldeau2000social,brockmann1979kleptoparasitism}. Free-riding refers to individuals exploiting collective efforts or shared resources without contributing their fair share. In team projects, for instance, free-riders neglect their assigned tasks to avoid costs or risks while still benefiting from the project's overall success \cite{holmstrom1982moral,bolton2004contract}. 
This phenomenon is also remarkably prevalent in animal foraging contexts, where individuals opportunistically engage in scrounging or kleptoparasitism, feeding off prey discovered or captured by others  \cite{cj1984evolution,packer1988evolution,brockmann1979kleptoparasitism,pitcher1993functions,Yossi,giraldeau1994test,di2001social}.

 The framework of {\em Producer-Scrounger (PS)} games is a widely used mathematical framework for studying free-riding  in foraging contexts \cite{giraldeau2000social,giraldeau2008social,vickery_producers_1991, giraldeau_food_1999,Yossi}. In a PS game, players are faced with a choice between two strategies: producer and scrounger. The interpretation of these strategies varies according to the context, but generally, a producer invests efforts in order to produce or find more resources, whereas a scrounger invests less in  producing or finding resources, and instead relies more on exploiting resources produced or found by others.  
Based on the rules of the particular PS game, specifying the production  and  rewarding mechanisms, each animal chooses a strategy and the system is assumed to converge into an equilibrium state, where each animal cannot improve its own calorie intake by changing its strategy \cite{smith_evolution_1982}.

This paper examines the impact of individual production capacity on the resulting payoffs in equilibrium configurations. The first  PS game we consider  aims to model a scenario consisting of a group of foraging animals, with each animal striving to maximize its own food intake. 
Intuitively, as long as the group size remains unchanged, one may expect that, even if it may trigger more opportunistic behavior \cite{callaway2002positive,brooker1998balance}, increasing food abundance should ultimately improve consumption rather than diminish it. 
Likewise, within a productivity-based reward system in a company, one may expect that enhancing individual productivity levels would boost group productivity and subsequently increase workers' payoffs, despite a possible increase in free-riding behavior. However, our findings uncover a more nuanced reality, unveiling a remarkably pronounced detrimental effect of free-riding behavior and emphasizing that the existence of such a positive correlation between individual productivity and payoffs is strongly contingent on the specific characteristics of the setting. 

\subsection*{Results}

We  investigate two types of PS games: a {\em foraging game} involving animals searching for food and a {\em company game} involving a group of workers in a company. Our main objective is to analyze the effects of changes in individual production capabilities on players' payoffs, evaluated at equilibrium configurations. To facilitate comparisons across different parameter settings, we ensure that the games we examine have unique
equilibrium configurations (see SI, \Cref{lem:ESS_sufficient}). We say that a PS game exhibits a {\em Reverse-Correlation} phenomenon if an increase in individuals' production capacities leads to a decrease in the players' payoff, when evaluated at equilibrium configurations (see \nameref{sec:methods}).

We begin with the Foraging game, which is a generalization of the classical PS game  in  \cite{vickery_producers_1991, giraldeau_food_1999}. The main difference with the classical model is that our model considers two types of food, instead of a single type as previously assumed.

\paragraph{The Foraging game.}  
To illustrate our model, consider a scenario involving a group of animals engaged in fruit picking from trees (see \Cref{fig:fruit_story}). Each animal aims to maximize its fitness, which is determined by the amount of food it consumes. The trees in this scenario contain both low-hanging fruit, accessible to both producers and scroungers, and high-hanging fruit, which can only be reached by producers. When an animal picks fruit, it retains a portion for its own consumption (let's say 70\%), while the remaining fruit falls to the ground. Scroungers, instead of picking high-hanging fruit, focus on scanning the ground for fallen fruit. Fallen fruit is distributed equally among all scroungers and the animal that originally obtained it.

More precisely, consider $n\geq 2$ animals, where each of which needs to choose to be either a {\em producer} or a {\em scrounger}. 
We assume that a producer finds an amount of food corresponding to $F_\producer = 1+\gamma$ calories, where, adhering to the trees example above, 1 corresponds to the amount of high-hanging fruit and $\gamma$ is a parameter that governs the animal's access to  low-hanging fruit. In contrast, a scrounger directly finds only low-hanging fruit, corresponding to $F_\scrounger = \gamma$ calories. After finding food (of any type) consisting of $F$ calories, the animal (either producer or scrounger) consumes a fraction~$s \in [0,1]$ of what it found  (called the ``finder's share'') and shares the remaining $(1-s)F$ calories equally with all scroungers. See \Cref{fig:foraging_schematic} for a schematic illustration of structure of the foraging game.

The {\em payoff} of a player is defined as the capacity of its calorie intake.
Hence, for each $0\leq k\leq n$, the payoff of each pure strategy in the presence of exactly~$k$ producers in the population is:
\begin{equation} \label{eq:producer_scrounger_payoff}
    \pi_\producer^{(k)} = s (1+\gamma) + (1-s) \frac{1+\gamma}{1+n-k}, \mbox{~~~and~~~}
    \pi_\scrounger^{(k)} = \gamma + k(1-s) \frac{1+\gamma}{1+n-k},
\end{equation} 
where the second equation follows since scrounger-to-scrounger interactions compensate each-others, and hence, can be ignored in the expression of the payoff. Note that the classical model \cite{vickery_producers_1991, giraldeau_food_1999} is retrieved by setting $\gamma=0$, which essentially implies that there is only one type of food. 

We study what happens to the payoffs of players at equilibria configurations, denoted by $\pi_\star$,  as we let $\gamma$  increase (see \nameref{sec:methods}). This increase aims to capture the case that the low-hanging fruit becomes more abundant in the environment.

Note that for each fixed $k$, both $\pi^{(k)}_\producer$ and $\pi^{(k)}_\scrounger$ are  increasing in $\gamma$. Hence, simply increasing $\gamma$, without changing the strategy profile, necessarily results in improved payoffs. However, allowing players to modify their strategies after such a change may potentially lead to enhanced scrounging at equilibrium, which can have a negative impact on the payoffs. Nevertheless,  as mentioned earlier, one might expect that this negative effect would be outweighed by the overall improvement in fruit availability, resulting in an increase in consumption rather than a decrease.
This intuition becomes apparent when comparing the scenarios with $\gamma=0$ and $\gamma=1$. As $\gamma$ increases from 0 to 1, we can expect an increase in the proportion of scroungers due to the rising ratio of $F_\scrounger/F_\producer=\gamma/(1+\gamma)$. However, even if the system with $\gamma=1$ ends up consisting entirely of scroungers, the average food consumption of a player (which equals 1) would still be at least as large as that of any strategy profile in the $\gamma=0$ case. Nonetheless, as shown here, upon closer examination within the interval $\gamma\in[0,1]$, a different pattern is revealed.  

We combine simulations (\Cref{fig:foraging}) with mathematical game-theoretical analysis (\Cref{thm:foraging_braess} in SI) to disclose a Reverse-Correlation phenomenon in the Foraging game. Specifically, for the case of $n=3$ players, we prove in \Cref{thm:foraging_braess} that for any finder's share $s<1/2$, there exists an interval of values for $\gamma$ over which the Reverse-Correlation phenomenon occurs. 
In the case of $2$ players, the Reverse-Correlation phenomenon does not happen over an interval, and 
instead, we prove in \Cref{thm:foraging_braess} 
 that there exists a critical value of $\gamma$ at which $\pi_\star$ decreases locally.
 
In our simulations, which focus on $n=4$ players, a noticeable decline is observed in the payoffs at equilibrium  as $\gamma$ increases over a relatively large sub-interval of $[0,1]$ (\Cref{fig:foraging_2D,fig:foraging_3D}). 

\begin{figure} [htbp]
    \centering
    
    \begin{subfigure}{\textwidth}
        \centering
        \begin{subfigure}{.45\textwidth}
            \centering
            \includegraphics[width=\linewidth]{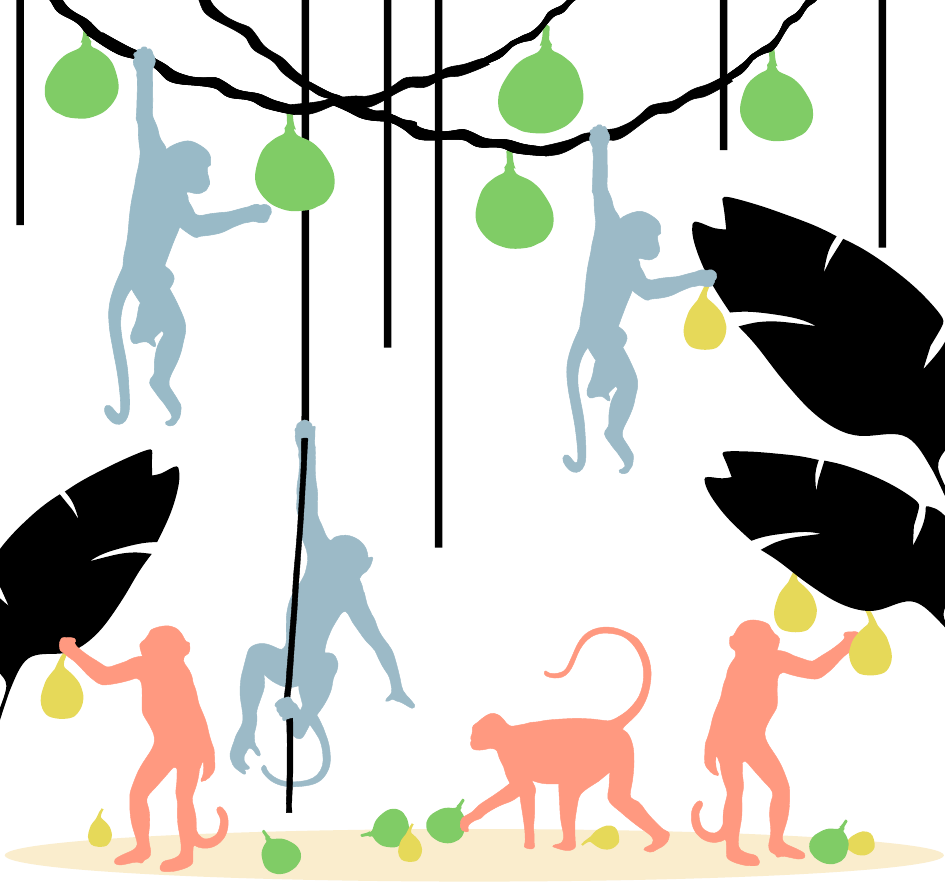}
            \caption{An illustration of the underlying scenario.}
            \label{fig:fruit_story}
        \end{subfigure}
        \hfill
        \begin{subfigure}{.45\textwidth}
            \centering
            \includegraphics[width=\linewidth]{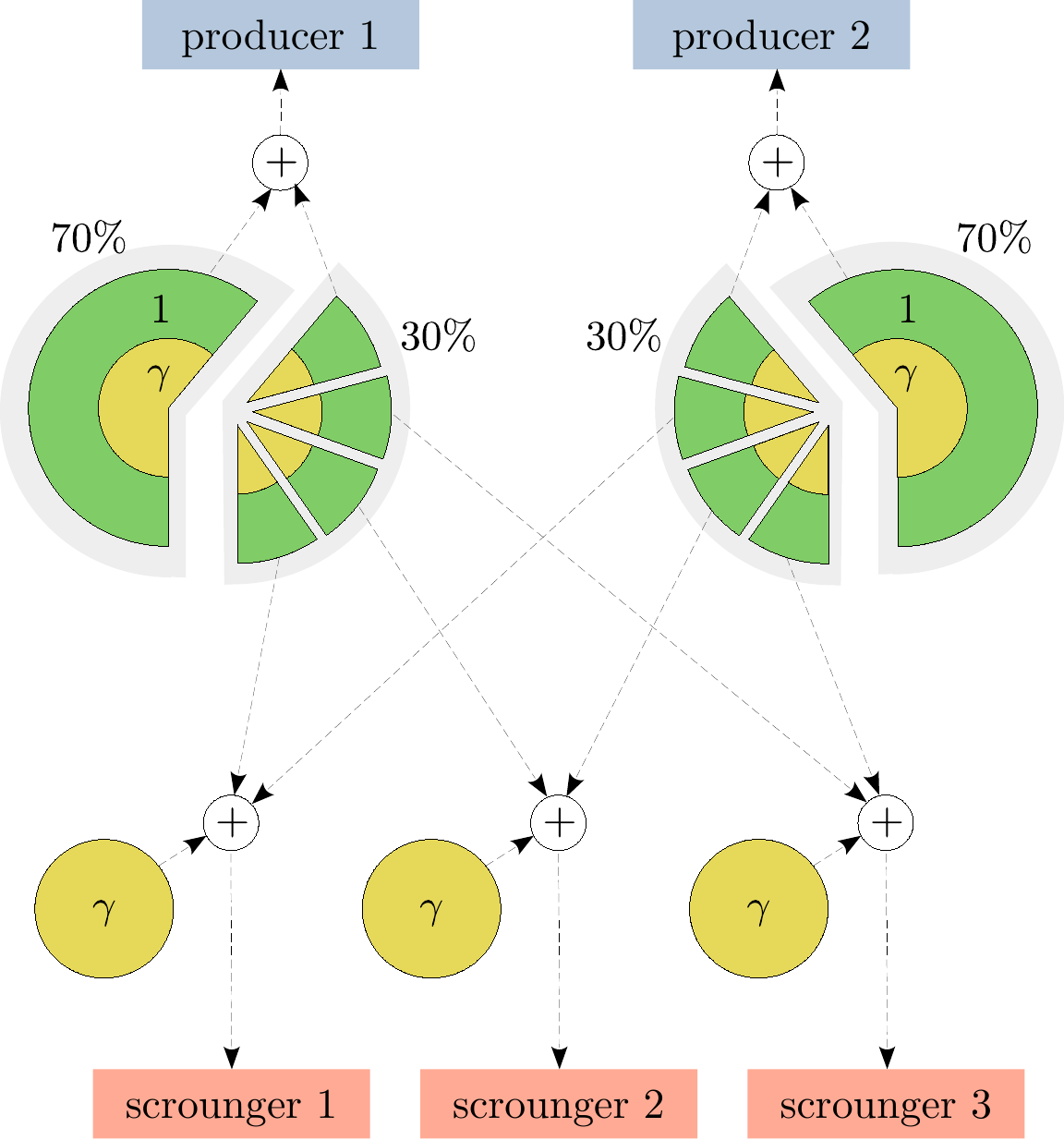}
            \caption{Schematic structure of the foraging game.}
            \label{fig:foraging_schematic}
        \end{subfigure}
        
    \end{subfigure}
    \vfill
    \begin{subfigure}{\textwidth}
        \centering
        \begin{subfigure}{.45\textwidth}
            \centering
            \includegraphics[width=\linewidth]{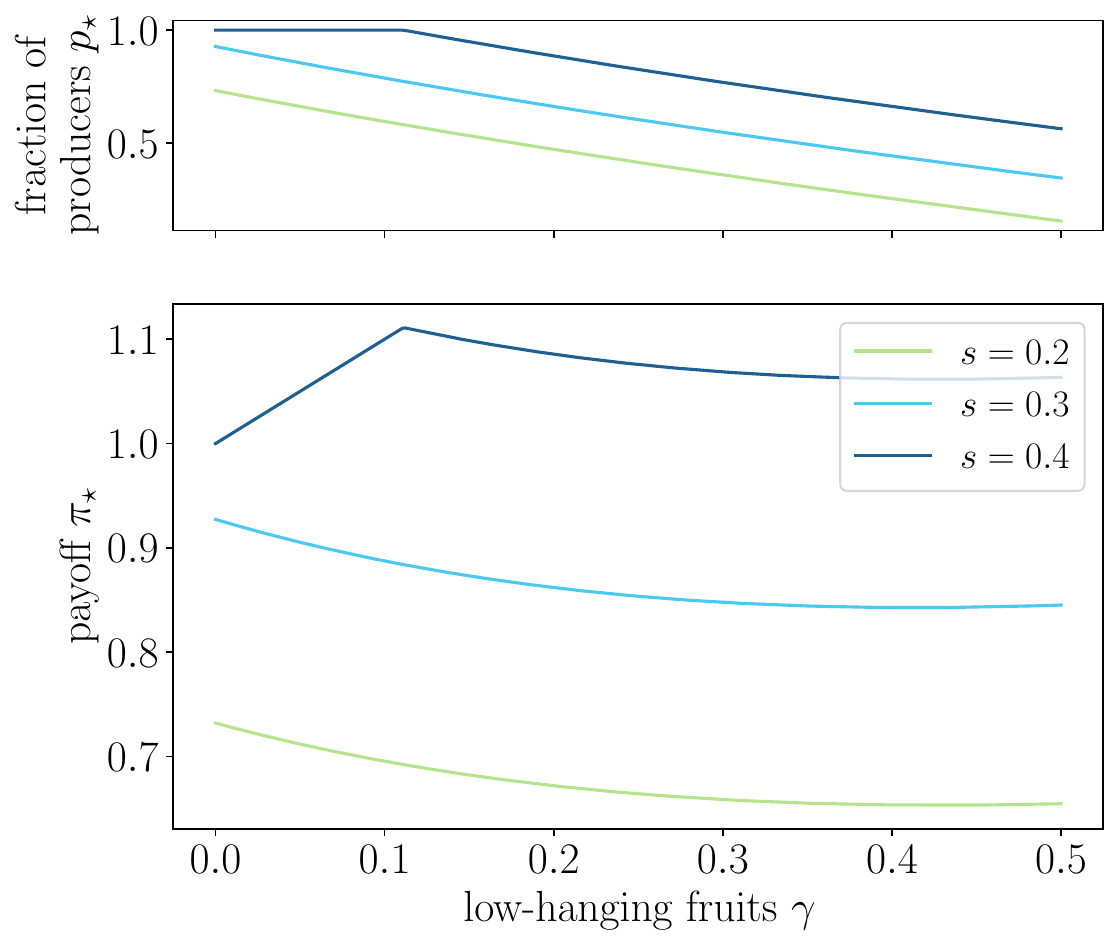}
            \caption{Payoff at equilibrium~$\pi_\star$ ($n=4$).}
            \label{fig:foraging_2D}
        \end{subfigure}
        \hfill
        \begin{subfigure}{.45\textwidth}
            \centering
            \includegraphics[width=\linewidth]{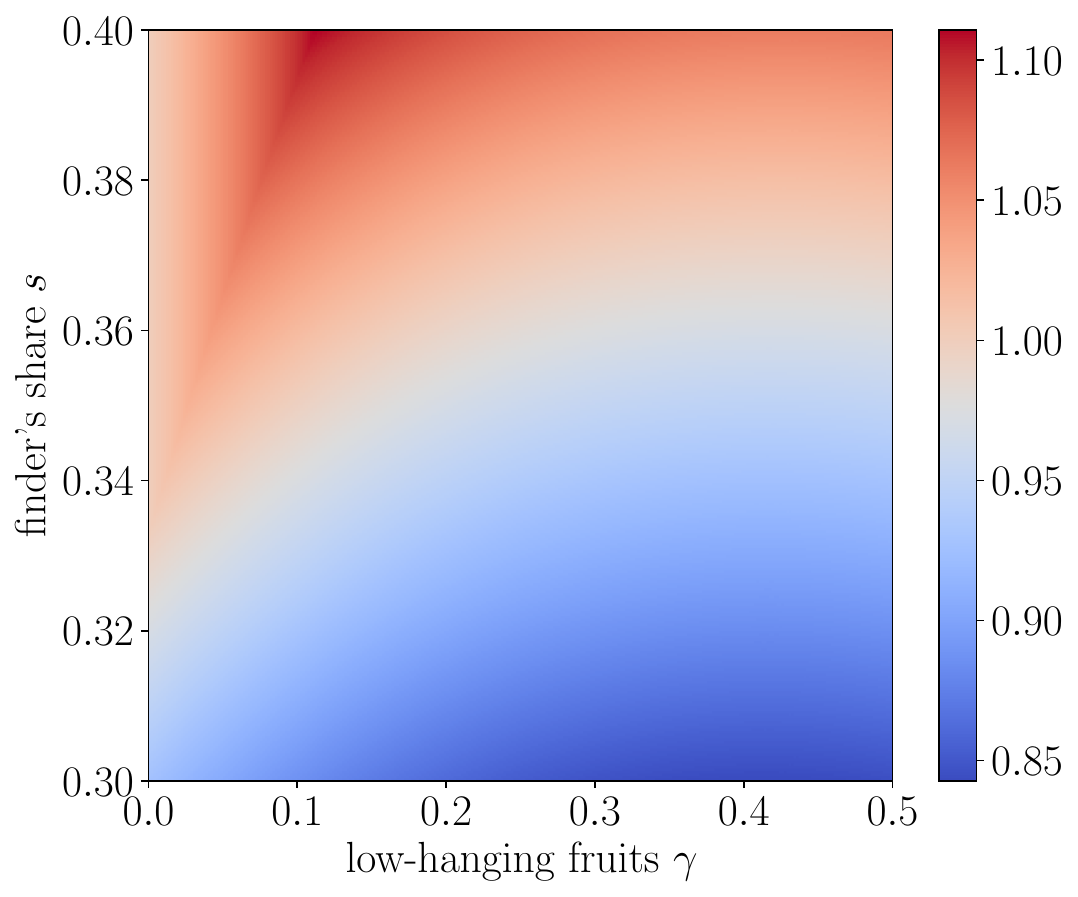}
            \caption{Payoff at equilibrium~$\pi_\star$ ($n=4$).}
            \label{fig:foraging_3D}
        \end{subfigure}
    \end{subfigure}
    
    \caption{
        {\bf The Foraging Game.}
        {\bf (a)} An illustration of the biological scenario that serves as an inspiration for the game. Producers (in blue) can grab high-hanging fruits (in green) as well as low-hanging fruits (in yellow). Scroungers (in red) are restricted to low-hanging fruits. Whenever some fruit is picked, a fraction is eaten directly by the animal, and the remaining falls to the ground, and is then consumed by the animal and the scroungers.
        {\bf (b)} The definition and key elements of the game are illustrated on a group of 5 animals --- 2 producers and 3 scroungers. Each scrounger finds food consisting of $\gamma$ calories (yellow, corresponding to low-hanging fruit) and each producer finds food of $1+\gamma$ calories, where 1 (green) corresponds to the high-hanging fruit and $\gamma$ to low-hanging fruit. A portion of $s=70\%$ is directly consumed by the finding animal, and the remaining $1-s=30\%$ portion (corresponding to fruit that fell on the ground) is equally shared between the animal and all scroungers. Scrounger-to-scrounger food exchanges cancel one another and their representations are hence omitted. 
        {\bf (c)} $n=4$. The probability of being a producer $p_\star(\gamma)$ and the payoff $\pi_\star(\gamma)$, at equilibria, for three values of the finder's share $s$. For each of these $s$ values, $\pi_\star$ is decreasing over a large interval of $\gamma$, effectively illustrating the Reverse-Correlation phenomenon.
        {\bf (d)} $n=4$. The relationship between the payoff at equilibrium $\pi_\star$ (color scale), the abundance of low-hanging fruits $\gamma$, and the finder's share~$s$. Once again, it illustrates the Reverse-Correlation phenomenon, while highlighting its independence from specific values of $s$.
    }
    \label{fig:foraging}
\end{figure}

\paragraph{The Company game.} We consider a PS game aiming to model a scenario with a group of $n\geq 2$ workers of equal capabilities who collaborate to produce a common product for a company. (Alternatively, by replacing the salary received by a worker with a notion of prestige, the game can also capture a scenario where a group of researchers collaborate in a research project.)

Each worker is assigned a specific part of the project and can choose between two pure behavioral strategies. 
A {\em producer} pays an {\em energetic cost} of $c>0$ units and with probability $p$  produces a product of quality $\gamma$ (otherwise, with probability $1-p$, it produces nothing). In contrast,  a {\em scrounger} pays no energetic cost and with probability $p$ produces a product of lower quality $\gamma'=a\gamma$ for some given $0\leq a<1$ (similarly, with probability $1-p$, nothing is produced).
Let $I=\{1,2,\ldots,n\}$, and let $q_i$ denote the quality of the product made by worker $i$, for $i\in I$, with $q_i=0$ if no product is made by this player. We define the {\em total production} as: 
\begin{equation}\label{eq:total}
    \Gamma=\sum_{i\in I} q_i. 
\end{equation}
We assume that the salary $\sigma_i$ of player $i$ is proportional to a weighted average between the quality of the products made by the workers, with more weight given to $q_i$.
If fact, by appropriately scaling the system,  we may assume without loss of generality that the salary is equal to this weighted average. Formally, we set:
\begin{equation}\label{eq:salary}
\sigma_i= s q_i + \frac{1-s}{n-1} \sum_{j\in I\setminus{i}} q_j, 
 \end{equation}
 for some $s \in [1/n,1]$.  
Note that $s=1$ implies that the salary each worker receives is identical to the quality of its own production, and $s=1/n$ represents equally sharing the quality of the global product between the workers. 

Next, we aim to translate the income salary of a player into his {\em payoff} using a {\em utility function}, denoted by $\phi(\cdot)$.
These quantities are expected to be positively correlated, however, the correlation may in fact be far from linear. Indeed, this is supported by the seminal work by Kahneman and Deaton \cite{kahneman2010high} which found that the emotional well-being of people increases relatively fast as their income rises from very low levels, but then levels off at a certain salary threshold. To capture such a correlation, we assume that $\phi$ is both monotonically non-decreasing, concave and bounded. 
In addition, the payoff of worker $i$ is decreased by its energetic investment. Finally, 
\begin{equation}\label{eq:researcher_simple}
   \pi_i:= \phi(\sigma_i)-c_i,
\end{equation}
where the energetic investment $c_i=c>0$ if $i$ is a producer and $c_i=0$ if $i$ is a scrounger. See \Cref{fig:company_payoff_mechanism} for an illustration of the semantic structure of the game.

The question of whether or not the system incurs a Reverse-Correlation phenomenon turns out to depend on the model's parameters, and, in particular, on the function $\phi(x)$. For example, when $\phi : x \mapsto x$ (i.e., the case that the salary is converted entirely into payoff), there is no Reverse-Correlation phenomenon (see Section \ref{observation} in SI). 
However, for some concave and bounded functions $\phi(x)$ the situation is different. 

We combine mathematical analysis with  simulations  
considering the function (see inset in \Cref{fig:company_payoff_mechanism}): \[\phi(x) = 1-\exp(-2 \, x).\] 
Our mathematical analysis proves the presence of a Reverse-Correlation phenomenon for the case of two workers (Theorem \ref{thm:company_braess} in the SI). Interestingly, this result holds for every~$s < 1$, demonstrating that the Reverse-Correlation phenomenon can occur even when the payoffs of individuals are substantially biased towards their own production compared to the production of others. 
Our simulations consider the case of $n=4$ workers and reveal (Figure \ref{fig:company}) that for certain parameters, letting $\gamma$ increase over a range of values results in a reduction in payoffs in equilibrium, thus indicating a Reverse-Correlation phenomenon. Moreover, as $\gamma$ increases over a range of values we also observe 
a substantial reduction in total production at equilibria. 

While the general shape of the utility function $\phi(x)= 1-\exp(-2 \, x)$ is justifiable, the  function itself was chosen somewhat arbitrarily. To strengthen the generality of our results, we  also provide in the SI (\Cref{sec:alternative_company_RC}) simulations supporting the Reverse-Correlation phenomenon  under another type of non-decreasing, concave, and bounded  utility function, specifically, \[\phi(x)= \min(1,x).\]  

\begin{figure} [htbp]
    \centering
    
    \begin{subfigure}{\textwidth}
        \centering
        
        \begin{subfigure}{.47\textwidth}
            \centering
            \includegraphics[width=\linewidth]{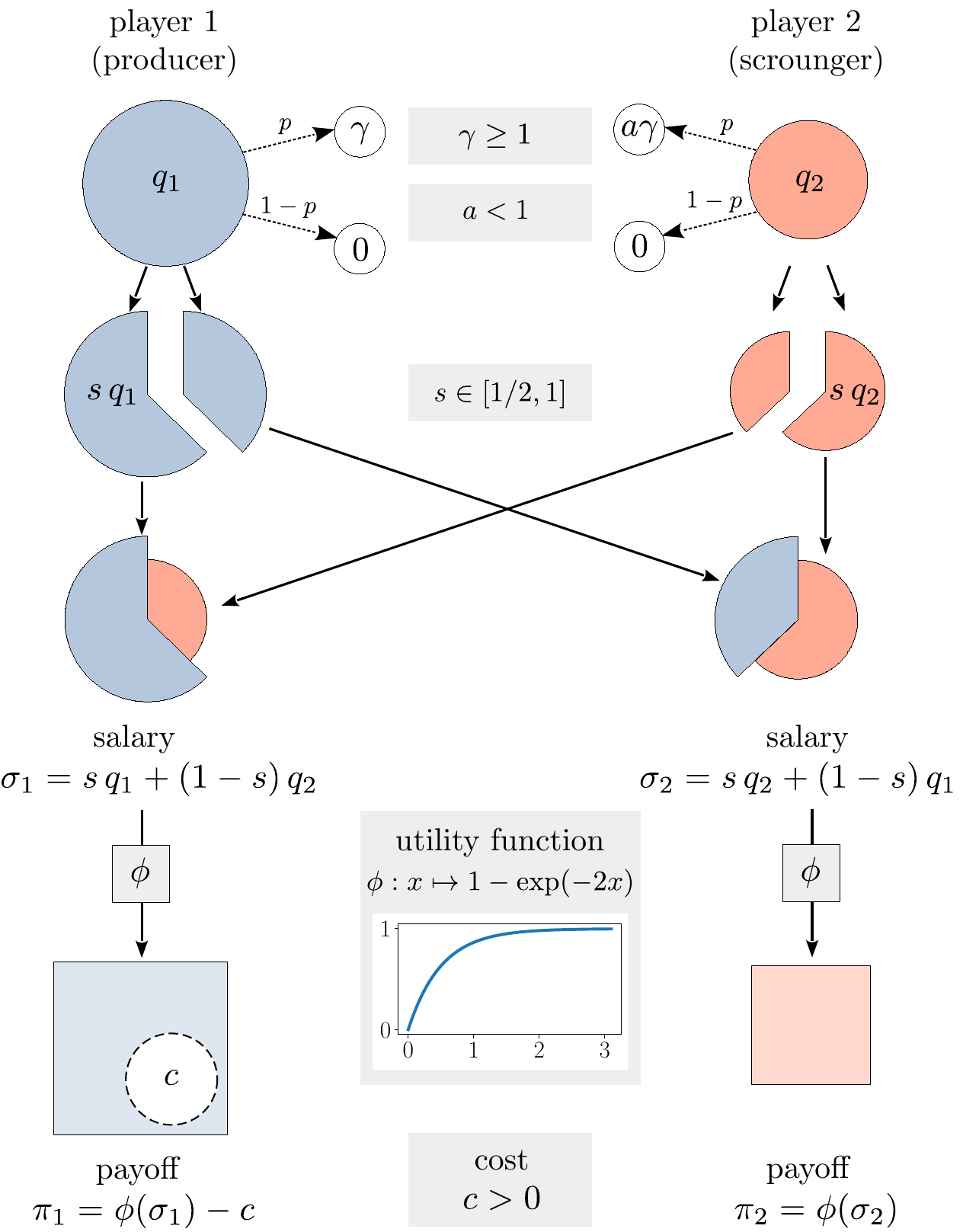}
            \caption{Schematic structure of the Company game.}
            \label{fig:company_payoff_mechanism}
        \end{subfigure}
        \hfill
        \begin{subfigure}{0.45\textwidth}
            \centering
            \includegraphics[width=\linewidth]{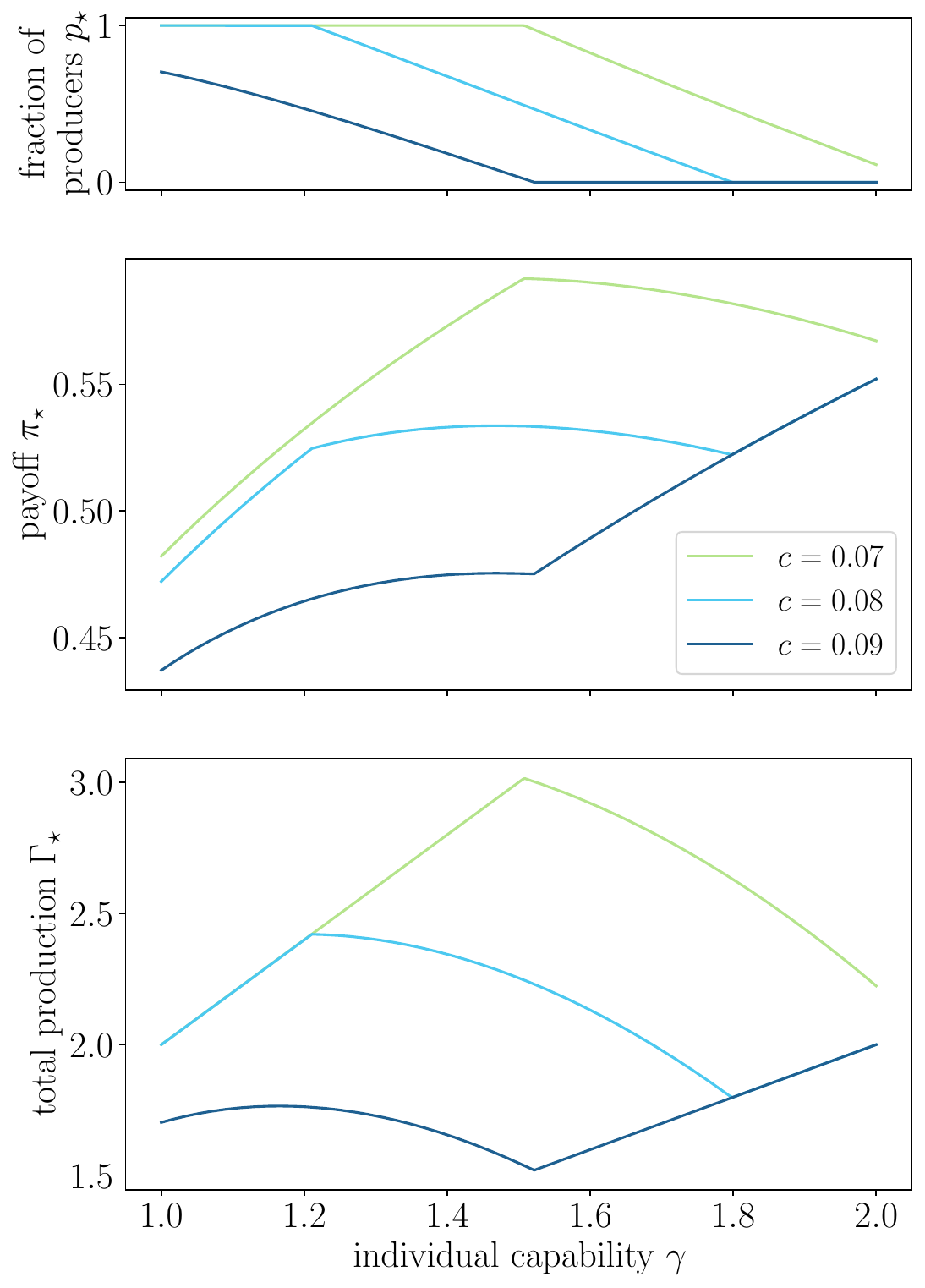}
            \caption{Payoffs and Total Production.}
            \label{fig:company_2D_payoffs}
        \end{subfigure}

    \end{subfigure}
    \vfill
    \begin{subfigure}{\textwidth}
        \centering
        \begin{subfigure}{0.45\textwidth}
            \centering
            \includegraphics[width=0.925\linewidth]{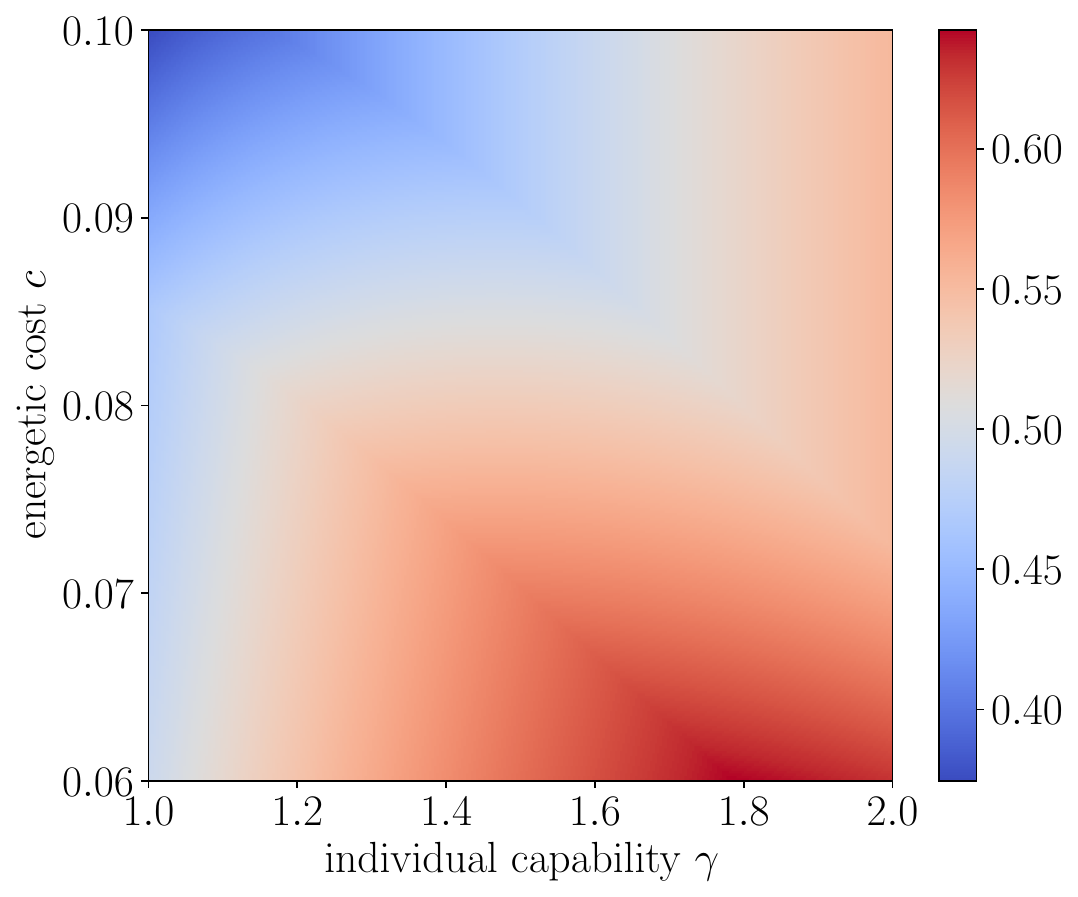}
            \caption{Payoffs.}
            \label{fig:company_3D_payoffs}
        \end{subfigure}
        \hfill
        \begin{subfigure}{0.45\textwidth}
            \centering
            \includegraphics[width=0.925\linewidth]{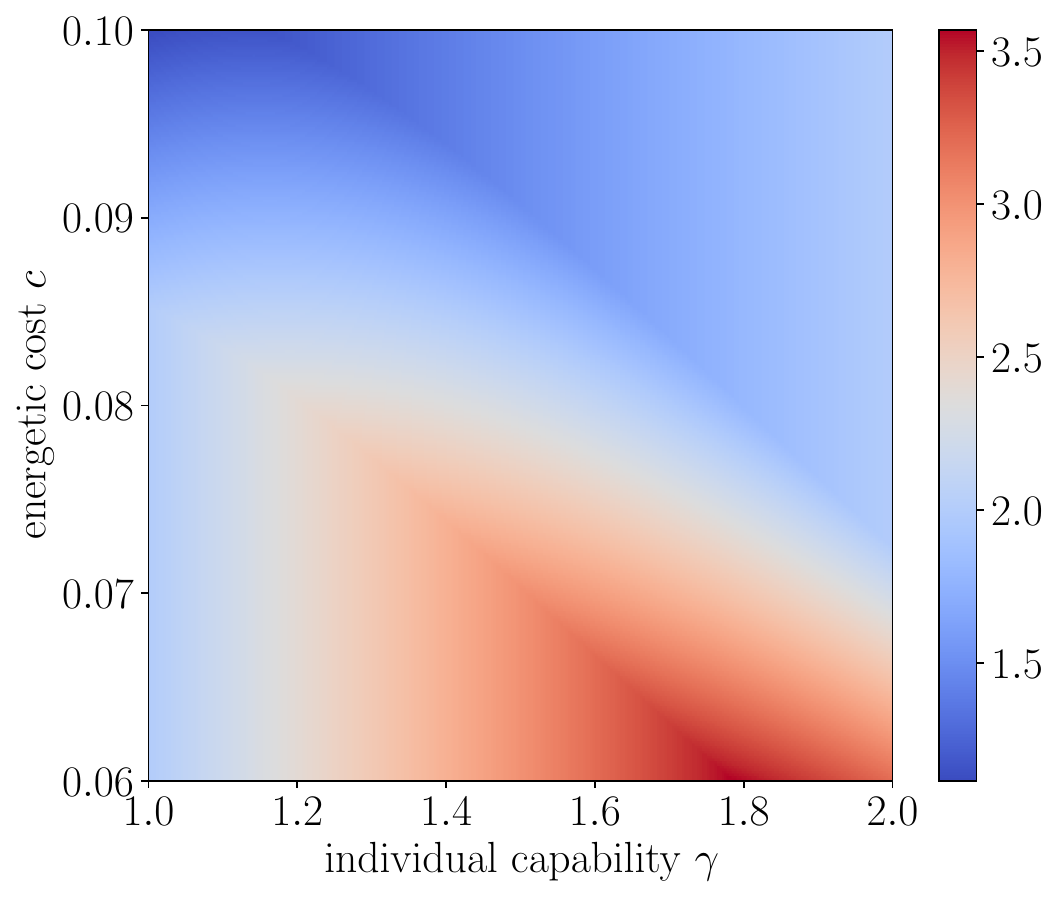}
            \caption{Total Production.}
            \label{fig:company_3D_totalProduction}
        \end{subfigure}
    \end{subfigure}
    
    \caption{{\bf The Company game}.
 {\bf (a)} Illustration of the game's definition for the case of 2 players:  a producer and a scrounger. 
The production $q_i$ of player $i$ is 0 with probability $1-p$, and otherwise, it is $\gamma$ if the player is a producer, and $a\gamma<\gamma$, otherwise. The salary $\sigma_i$ is a weighted average of the production of both players with more weight given to $q_i$. The utility function $\phi$ (inset) maps salary into the payoff, from which the energetic cost is withdrawn.
Finally, $\pi_i = \phi(\sigma_i) - c_i$, where $c_i$ is equal to $c$ if $i$ is a producer, and 0, otherwise.
{\bf (b)---(d)} Simulating a scenario with $n=4$ players, with the utility function $\phi: x \mapsto 1-\exp(-2x)$, and assuming $a=p=\frac{1}{2}$, and $s=0.6$.
{\bf (b)} The graph depicts the payoff $\pi_\star(\gamma)$ and the total production $\Gamma_\star(\gamma)$, as well as the probability of being a producer $p_\star(\gamma)$ at equilibria, as a function of individual capabilities $\gamma$ and for several values of $c$. It displays the Reverse-Correlation phenomenon over a certain interval that depends on $c$.
{\bf (c)} and {\bf (d)} The relationship between {\bf (c)} payoff $\pi_\star$,  and {\bf (d)} total production $\Gamma_\star$ at equilibrium, as a function (color scale) of individual capabilities $\gamma$ and the cost $c$ for production.} 
    \label{fig:company}
\end{figure}

\paragraph{A necessary condition.}
Finally, we identify  a necessary condition for the emergence of a Reverse-Correlation phenomenon in arbitrary PS models. Specifically, we prove (SI, \Cref{sec:necessary}) that a Reverse-Correlation phenomenon can  occur only if the definition of the producer's payoff is sensitive to the fraction of scroungers in the population.

An interesting consequence of this condition is that a seemingly minor change in the definition of the Foraging game can prevent the occurrence of the Reverse-Correlation phenomenon. Recall that in this game, when an animal finds food, it consumes a fraction $s$ of it (the finder's share), and the remaining $1-s$ fraction falls to the ground and is then equally shared between the animal and all scroungers. 
If the game is changed such that when a producer finds food, it only consumes the finder's share and does not eat at all from the food that falls on the ground (i.e., only the scroungers eat from it), then the game stops satisfying the aforementioned necessary condition. Indeed, in this case, the payoff of a producer would always be $1+\gamma$ irrespective of the number of scroungers.
Hence, the modified game does not exhibit a Reverse-Correlation phenomenon, regardless of the parameters involved.

\subsection*{Discussion}

In foraging contexts, it is commonly anticipated that an increase in food abundance would result in higher consumption, which, in turn, would lead to population growth over time.  In contrast, this paper introduces the intriguing possibility of a reversed scenario: that under certain producer/scrounger conditions, if animals have sufficient time to update their producer/scrounger strategy and reach a stable configuration before reproducing \cite{aplin2017stable,katsnelson2008early,morand2010learning}, then an increase in food abundance can paradoxically result in reduced consumption, which, in turn, can lead to a decline in population size! Note that this idea can also be viewed from the opposite perspective, namely, that by reducing food abundance, the inclination to scrounge can decrease, resulting in improved food consumption, ultimately leading to an increase in population size.  

The Reverse-Correlation phenomenon corresponds to a decrease in payoffs as underlying conditions improve. The counter-intuitive aspect of it stems from the fact that players aim to maximize their payoffs, yet when conditions improve, they are driven to perform worse.
 Another measure of interest is the total production, defined as the sum of production over all players (\Cref{eq:total}). Observe that in the Foraging game, since the animals eventually consume all food found by the group, the total production (i.e., the total food found) at equilibrium is proportional to the payoff $\pi_\star$, and hence their dynamics are similar.
This implies that whenever an increase in $\gamma$ results in a decrease in payoff at equilibrium  (indicating a Reverse-Correlation phenomenon), the same increase in $\gamma$ also leads to a decrease in total production at equilibrium. In contrast, in the Company game,  production is not fully represented in the payoffs, since some of it is ``lost'' when translating salaries into utilities. Additionally, the distinction between payoffs and production is further emphasized due to the energetic cost incurred by producers, which is reflected in their payoffs. Despite this distinction, as observed in \Cref{fig:company_2D_payoffs,fig:company_3D_totalProduction},
the measure of total production also exhibits a decrease across a range of $\gamma$ values. This phenomenon may carry particular importance for system designers, such as the company's principal, as it challenges a fundamental assumption underlying bottom-up approaches, namely, that as long as the system naturally progresses without external disruptions, improving individual performances should lead to enhanced group performances.

We demonstrated the Reverse-Correlation phenomenon on two basic game theoretical models. As evident by these games, the occurrence of this counter-intuitive phenomenon is highly contingent on the specific details of the game. 
For example, the Foraging game considers two types of food: low-hanging and high-hanging fruit (instead of just one type as consider in the classical game in \cite{vickery_producers_1991, giraldeau_food_1999}). Only producers have access to high-hanging fruit, while both producers and scroungers can access low-hanging fruit. Similarly to the classical model, when an animal finds food, it consumes a portion $s$ of it and the remaining $1-s$ portion is equally shared between this animal and all scroungers. The Reverse-Correlation phenomenon emerges as  the abundance of low-hanging fruit increases. However, as we showed, if one modifies the model so that the remaining $1-s$ portion is shared only between the scroungers, then the system no longer exhibits a Reverse-Correlation phenomenon. Hence, while at first glance this change may appear minor, it has a profound impact on the dynamics. 
In the Company game, a key aspect of the model concerns the choice of the utility function, which captures the relationship between salary and payoff. Inspired by the work of Kahneman and Deaton \cite{kahneman2010high}, we focused on non-decreasing, concave, and bounded utility functions. Within this family of functions, we identified two that exhibit a Reverse-Correlation phenomenon. However, we note that  not all utility functions in this family enable this phenomenon.

In conclusion, this paper uncovers a counter-intuitive phenomenon that can arise in productive group contexts involving rational players. It reveals that under certain conditions, increasing individual production efficiency can paradoxically lead to diminished payoffs and overall group production, due to a significant rise in free-riding behavior. These findings provide valuable insights into the complex dynamics at play, underscoring the intricate relationship between individual and group performances, as well as the detrimental impact of free-riding behavior. Moreover, our results highlight the nuanced consequences of contextual factors in understanding and predicting the impact of increased (or decreased) resource availability on both individual and collective outcomes.
 
\subsection*{Methods}\label{sec:methods}
We consider two types of PS models, for which we combine analytic investigations with computer simulations. In both models, we assume that both producers and scroungers are able to produce, but that producers are expected to produce more.
In our models, the payoffs and total production are positively correlated with the number of producers.
We consider a parameter $\gamma$ that is positively correlated to the expected production capacities of both producers and scroungers. To check what happens as individual capabilities improve, we increase $\gamma$ and observe how the payoff and total production measures change, for configurations at equilibria. 

We focus on the strong definition of equilibria, known as Evolutionary Stable Strategy (ESS), using the standard definition as introduced by Maynard Smith and Price \cite{smith_evolution_1982}. 
Specifically, given a PS game, let $\payoff{q}{p}$ denote the expected payoff of a player if it chooses to be a producer with probability $q$, in the case that all $n-1$ remaining players are producers with probability $p$.
We say that $p_\star \in [0,1]$ is an ESS if and only if for every~$q \in [0,1]$ such that $q \neq p_\star$,
\begin{itemize}
    \item[(i)] either $\payoff{p_\star}{p_\star} > \payoff{q}{p_\star}$,
    \item[(ii)] or $\payoff{p_\star}{p_\star} = \payoff{q}{p_\star}$ and $\payoff{p_\star}{q} > \payoff{q}{q}$.
\end{itemize}
To be able to compare instances with different parameters, we make sure that for every value of $\gamma$, the game we consider always has a unique ESS, termed $p_\star(\gamma)$. In such a case, we write $\pi_\star(\gamma) = \pi_{p_\star(\gamma),p_\star(\gamma)}$ the payoff at the ESS, and omit the parameter $\gamma$ when clear from the context. 

In our rigorous analysis, presented in the SI, we prove the existence and uniqueness of the ESS, for the corresponding scenarios. 
To determine the ESS in our simulations, we utilize simple procedures that take the values of $p$ and $q$ as inputs and calculate $\payoff{q}{p}$.
Then, we search for the specific value of $p$ that satisfies (i) $\payoff{1}{p} = \payoff{0}{p}$, (ii) for every~$q<p$, $\payoff{1}{q} > \payoff{0}{q}$ and (iii) for every~$q>p$, $\payoff{1}{q} < \payoff{0}{q}$, which together are sufficient conditions for $p$ to be the unique ESS (see SI, \Cref{lem:ESS_sufficient}). 

Both the code used in the simulations and the code employed to generate the figures were implemented in Python. For further details and access to the code, please refer to \cite{Vacus_Figure_Generation_Code_2023}. 

We say that the system incurs a {\em Reverse-Correlation} phenomenon if increasing $\gamma$ over a certain interval yields decreased payoff when evaluated at (the unique) ESS. In other words, this means that $\pi_\star(\gamma)$ is a decreasing function of $\gamma$ over this interval.

\paragraph{Acknowledgements.} The authors would like to thank Yossi Yovel, Ofer Feinerman, Yonatan Zegman and Yannick Viossat for helpful discussions.

\bibliographystyle{unsrt}
\bibliography{ref}

\clearpage

\centerline{\huge Supplementary Information}
\vspace{70pt}

\section{Uniqueness of ESS}\label{sec:ESS}

The following sufficient condition for the existence and uniqueness of an ESS is well-known. We state and prove it below for the sake of completeness.
\begin{lemma} \label{lem:ESS_sufficient}
    If $p_\star \in [0,1]$ is such that (i) $\payoff{1}{p_\star} = \payoff{0}{p_\star}$, (ii) for every~$q<p_\star$, $\payoff{1}{q} > \payoff{0}{q}$ and (iii) for every~$q>p_\star$, $\payoff{1}{q} < \payoff{0}{q}$, then $p_\star$ is a unique ESS.
\end{lemma}
\begin{proof}
    By assumption (i), we have for every~$q \in [0,1]$:
    \begin{equation*}
        \payoff{q}{p_\star} = q \, \payoff{1}{p_\star} + (1-q) \, \payoff{0}{p_\star} = p_\star \, \payoff{1}{p_\star} + (1-p_\star) \, \payoff{0}{p_\star} = \payoff{p_\star}{p_\star}.
    \end{equation*}
    Thus, to show that $p_\star$ is an ESS, we need to check the second condition in the definition.
    We start by considering the case that~$q < p_\star$. By assumption (ii), it implies that $\payoff{1}{q} > \payoff{0}{q}$, so
    \begin{equation*}
        \payoff{p_\star}{q} = p_\star \, \payoff{1}{q} + (1-p_\star) \, \payoff{0}{q} > q \, \payoff{1}{q} + (1-q) \, \payoff{0}{q} = \payoff{q}{q}.
    \end{equation*}
    Similarly, in the case that~$q > p_\star$, assumption (iii) implies that $\payoff{1}{q} < \payoff{0}{q}$, so
    \begin{equation*}
        \payoff{p_\star}{q} = p_\star \, \payoff{1}{q} + (1-p_\star) \, \payoff{0}{q} > q \, \payoff{1}{q} + (1-q) \, \payoff{0}{q} = \payoff{q}{q}.
    \end{equation*}
    By the second condition in the definition, this implies that $p_\star$ is an ESS.
    
    Finally, we prove the unicity property.
    Let $p \neq p_\star$.
    If $p<p_\star$, then $\payoff{1}{p} > \payoff{0}{p}$ by assumption $(ii)$, and $p<1$. Therefore,
    \begin{equation*}
        \payoff{1}{p} > p \, \payoff{1}{p} + (1-p) \, \payoff{0}{p} = \payoff{p}{p}.
    \end{equation*}
    If $p>p_\star$, then $\payoff{1}{p} < \payoff{0}{p}$ by assumption $(iii)$, and $p>0$. Therefore,
    \begin{equation*}
        \payoff{0}{p} > p \, \payoff{1}{p} + (1-p) \, \payoff{0}{p} = \payoff{p}{p}.
    \end{equation*}
    In both cases, $p$ is not an ESS, which concludes the proof of \Cref{lem:ESS_sufficient}.   
\end{proof}

\section{Analysis of the Foraging game}\label{sec:foraging}
The goal of this section is to prove   \Cref{thm:foraging_braess}. Note that the theorem considers $n=2,3$. For the case of $3$ players, the theorem states that as long as the finder's share satisfies $s<1/2$, there exists an interval of values for $\gamma$ over which the Reverse-Correlation phenomenon occurs. In contrast, in the case of $2$ players, the Reverse-Correlation phenomenon does not happen over an interval, and 
instead, there exists a critical value of $\gamma$ at which $\pi_\star$ decreases locally. In fact, this happens even when the finder's share is close to 1.

\begin{theorem} \label{thm:foraging_braess}    
    Consider the Foraging game with $\gamma \geq 0$ and $s<1$.
    \begin{itemize}
        \item If $n=3$, then for every~$\gamma \geq 0$, there is a unique ESS.
        Moreover, for every~$s < 1/2$, there exist $\gamma_{\min}, \gamma_{\max} > 0$ such that  the payoff $\pi_\star(\gamma)$ (and hence, also the total production) at ESS is strictly decreasing in $\gamma$ on the interval~$[\gamma_{\min},\gamma_{\max}]$.
        
        \item If $n=2$, then for every~$\gamma \neq \gamma_s$, where $\gamma_s = \frac{1+s}{1-s}$, there is a unique ESS.
        Moreover, $\pi_\star(\gamma)$ is increasing on~$[0,\gamma_s)$ and on $(\gamma_s,+\infty]$.
        However, for every $\epsilon \in (0,1/2)$, $\pi_\star(\gamma_s-\epsilon)>\pi_\star(\gamma_s+\epsilon)$.
    \end{itemize}
\end{theorem}

\subsection{Proof of Theorem \ref{thm:foraging_braess}}

Towards proving the theorem, we first establish the following lemma, which quantifies the expected payoffs of the two pure strategies, conditioning on other agents choosing to be producers with probability $p$.  
\begin{lemma} \label{lem:expectation_payoff}
    For every $0\leq p<1$,
    \begin{equation*}
        \payoff{1}{p} = s F_\producer + (1-s)F_\producer \cdot \frac{1-p^n}{n(1-p)},
    \end{equation*}
    and
    \begin{equation*}
        \payoff{0}{p} = F_\scrounger + (1-s)F_\producer \cdot p \cdot \frac{n(1-p)+p^{n}-1}{n(1-p)^2}.
    \end{equation*}
    These expressions can be extended by continuity at $p=1$, giving
    $\payoff{1}{1} = F_\producer$
    and
    $\payoff{0}{1} = F_\scrounger + (n-1)(1-s)F_\producer/2$.
\end{lemma}
\begin{proof}
    Fix a player~$i$. Consider the case that Player~$i$ is a producer, and that each player $j\neq i$ is a producer with probability~$p$.
    Let $X_p$ be the random variable indicating the number of scroungers in the population. By \Cref{eq:producer_scrounger_payoff},
    \begin{equation*}
        \payoff{1}{p} = s F_\producer  + (1-s) F_\producer \cdot \bbE \pa{ \frac{1}{1+X_p} }.
    \end{equation*}
    By definition, $X_p \sim \calB(n-1,1-p)$. The first part of the claim,  concerning $\payoff{1}{p}$, now follows using \Cref{claim:expectation_binomial_1}, that implies that $\bbE (1/(1+X_p)) = \frac{1-p^n}{n(1-p)}$.
    
    Now, consider the case that Player~$i$ is a scrounger, and that each player $j\neq i$ is a producer with probability~$p$.
    Let $Y_p$ be the random variable indicating the number of producers in the population. By \Cref{eq:producer_scrounger_payoff},
    \begin{equation*}
        \payoff{0}{p} = F_\scrounger + (1-s) F_\producer \cdot \bbE \pa{ \frac{Y_p}{1+n-Y_p} }.
    \end{equation*}
    By definition, $Y_p \sim \calB(n-1,p)$. The second part of the claim,  concerning $\payoff{0}{p}$, now follows using \Cref{claim:expectation_binomial_3}, that implies that
    \begin{equation*}
        \bbE\pa{\frac{Y_p}{1+n-Y_p}} = \bbE\pa{\frac{Y_p}{2+(n-1)-Y_p}} = p \cdot \frac{n(1-p)+p^{n}-1}{n(1-p)^2}.
    \end{equation*}
    This completes the proof of Lemma \ref{lem:expectation_payoff}.
\end{proof}

In order to characterize the (unique) ESS,
we first define the following quantities:
\begin{equation*}
    A(\gamma) = \frac{n(F_\scrounger - s F_\producer)}{(1-s)F_\producer} = \frac{n(\gamma - s(1+\gamma))}{(1-s)(1+\gamma)}, \quad \gamma_1 = \frac{2}{(n-1)(1-s)}-1, \quad \gamma_2 = \frac{n}{(n-1)(1-s)}-1.
\end{equation*}

\begin{claim} \label{claim:simple_computations}
		We have $A(\gamma_1) = -n(n-3)/2$ and $A(\gamma_2) = 1$.
	\end{claim}
	\begin{proof}
		First, we rewrite
        \begin{equation} \label{eq:new_expression_A}
            A(\gamma) = \frac{n(\gamma - s(1+\gamma))}{(1-s)(1+\gamma)} = n \cdot \frac{\gamma(1-s) - s}{(1-s)(1+\gamma)}
            = n \cdot \pa{1- \frac{1}{(1-s)(1+\gamma)}}.
        \end{equation}
        Plugging in the definition of~$\gamma_1$ and $\gamma_2$, we obtain
        \begin{equation*}
            A(\gamma_1) = n \cdot \pa{1- \frac{1}{\frac{2}{n-1}}} = -\frac{n(n-3)}{2} \quad \text{ and } \quad A(\gamma_2) = n \cdot \pa{1- \frac{1}{\frac{n}{n-1}}} = 1,
        \end{equation*}
        as stated.
	\end{proof}

Next, for every~$\gamma$, the following result identifies the unique ESS.
\begin{lemma} \label{lem:equilibrium_condition}
    ~
    \begin{itemize}
        \item[(a)] For every $n \geq 2$, for every~$\gamma \in [0,\gamma_1) \cup (\gamma_2, +\infty)$, there is unique ESS, termed $p_\star(\gamma)$, that satisfies $p_\star(\gamma) = 1$ on $[0,\gamma_1)$ and $p_\star(\gamma) = 0$ on $(\gamma_2,+\infty)$.
        \item[(b)] for every $n \geq 3$, for every~$\gamma \in [\gamma_1,\gamma_2]$, there is unique ESS, termed $p_\star(\gamma)$. Moreover, $p_\star$ is continuously differentiable on $[\gamma_1,\gamma_2]$, $p_\star(\gamma_1) = 1$ and $p_\star(\gamma_2) = 0$.
    \end{itemize}
\end{lemma}
\begin{proof}
	Define the following function for $0\leq p<1$.
	\begin{equation*} 
		f(p) = \frac{1}{1-p} \pa{ \frac{1-p^n}{1-p} - np }.
	\end{equation*}
    
   We next identify $\lim_{p\rightarrow 1}f(p)$.
    \begin{observation} \label{claim:lim_f}
        Function $f$ can be extended to a continuous function at $p=1$ by setting $f(1) = -n(n-3)/2$.
    \end{observation}
    \begin{proof}
        Let $x = 1-p$. Using Taylor expansion at $x=0$, we have:
        \begin{align*}
            f(x) &= \frac{1}{x} \pa{ \frac{1-(1-x)^n}{x} - n(1-x) } = \frac{1}{x} \pa{ \frac{nx-\frac{n(n-1)}{2}\, x^2 + o(x^3)}{x} - n(1-x) } \\
            &= \frac{1}{x} \pa{ n \pa{1-\tfrac{n-1}{2} \, x + o(x^2)} - n(1-x) } = \frac{n}{x} \pa{ - \tfrac{n-3}{2} \, x + o(x^2) } = -\frac{n(n-3)}{2} + o(x).
        \end{align*}
     Therefore, $\lim_{p \rightarrow 1} f(p) = -\frac{n(n-3)}{2}$, which concludes the proof of the observation.
    \end{proof}
 
    To compute the ESS, we need to compare $\payoff{1}{p}$ and $\payoff{0}{p}$.
    \begin{claim} \label{claim:new_nash_condition}
       For every~$p \in [0,1]$, $\payoff{1}{p} > \payoff{0}{p} \iff f(p) > A(\gamma)$ and $\payoff{1}{p} < \payoff{0}{p} \iff f(p) < A(\gamma)$.
    \end{claim}
    \begin{proof}
        By \Cref{lem:expectation_payoff}, for every~$p \in [0,1]$,
        \begin{align*}
            \payoff{1}{p} > \payoff{0}{p} &\iff s F_\producer + (1-s)F_\producer \cdot \frac{1-p^n}{n(1-p)} > F_\scrounger + (1-s)F_\producer \cdot p \cdot \frac{n(1-p)+p^{n}-1}{n(1-p)^2} \\
            &\iff \frac{1-p^n}{1-p} - p \cdot \frac{n(1-p)+p^{n}-1}{(1-p)^2} > \frac{n(F_\scrounger - s F_\producer)}{(1-s)F_\producer}.
        \end{align*}
        By definition, the right hand side is equal to~$A(\gamma)$. Let us rewrite the left hand side:
        \begin{equation*}
            \frac{1-p^n}{1-p} - p \cdot \frac{n(1-p)+p^{n}-1}{(1-p)^2} = \frac{1}{1-p} \pa{ (1-p) \frac{1-p^n}{1-p} - np + p \frac{1-p^n}{1-p}} = \frac{1}{1-p} \pa{ \frac{1-p^n}{1-p} - np } = f(p),
        \end{equation*}
        which concludes the proof of the first equivalence in \Cref{claim:new_nash_condition}. The second equivalence is obtained similarly.
    \end{proof}
	\begin{claim} \label{claim:f_decreasing}
         $f$ is non-increasing in $p$.
		Moreover, if $n \geq 3$, then $f$ is strictly decreasing in $p$.
	\end{claim}
	\begin{proof}
        First, consider the case that $n=2$. Then,
        \begin{equation*}
            f(p) = \frac{1}{1-p} \pa{ \frac{1-p^2}{1-p} - 2p } = \frac{1}{1-p} \pa{ (1+p) - 2p } = 1,
        \end{equation*}
        so $f$ is non-increasing.
        Now, consider the case that $n \geq 3$.
        Let us write $f(p) = u(p)/v(p)$, with
        \begin{equation*}
            u(p) = \frac{1-p^n}{1-p} - np, \quad v(p) = 1-p.
        \end{equation*}
        We have
        \begin{equation*}
            u'(p) = \frac{-n p^{n-1}(1-p) + (1-p^n)}{(1-p)^2} - n, \quad v'(p) = -1,
        \end{equation*}
        so
        \begin{equation*}
            u'(p) \cdot v(p) = -n p^{n-1} + \frac{1-p^n}{1-p} - n(1-p), \quad u(p) \cdot v'(p) = - \frac{1-p^n}{1-p} + np.
        \end{equation*}
        Therefore,
        \begin{equation*}
            u'(p) \cdot v(p) - u(p) \cdot v'(p) = -n p^{n-1} + 2 \frac{1-p^n}{1-p} - n = 2 \frac{1-p^n}{1-p} - n(1+p^{n-1}).
        \end{equation*}
        Finally,
        \begin{equation} \label{eq:f_derivative}
            f'(p) = \frac{u'(p) \cdot v(p) - u(p) \cdot v'(p)}{v(p)^2} = - \frac{n(1-p)(1+p^{n-1}) - 2(1-p^n)}{(1-p)^3}.
        \end{equation}
        Let us define the ratio:
        \begin{equation*}
            g_0(p) = \frac{n(1-p)(1+p^{n-1})}{2(1-p^n)}.
        \end{equation*}
        Next, we  show that $g_0$ is strictly greater than $1$. 
        To this aim, we study $g_0$ by differentiating it several times. Define:
        \begin{equation*}
            g_1(p) = (n-1)p^{n-2}(1-p^2)+p^{2n-2}-1, \quad g_2(p) = 2p^n-np^2+n-2, \quad g_3(p) = -2np(1-p^{n-2}).
        \end{equation*}
        Since $n \geq 3$, 
        $g_3(p) < 0$.
        We have
        \begin{equation*}
            g_2'(p) = g_3(p) < 0,
        \end{equation*}
        so $g_2$ is strictly decreasing, and hence $g_2(p) > g_2(1) = 0$.
        We have
        \begin{equation*}
            g_1'(p) = (n-1)p^{n-3} g_2(p) > 0,
        \end{equation*}
        so $g_1$ is strictly increasing, and $g_1(p) < g_1(1) = 0$.
        Eventually, we have:
        \begin{align*}
            g_0'(p) &= \frac{n}{2} \cdot \frac{\pa{-1-p^{n-1}+(n-1)p^{n-2}(1-p)}\cdot(1-p^n)+(1-p)(1+p^{n-1})\cdot n\, p^{n-1}}{(1-p^n)^2} \\
            &= \frac{n}{2} \cdot \frac{ -1-p^{n-1}+(n-1)p^{n-2}(1-p)+p^n+p^{2n-1}-(n-1)p^{2n-2}(1-p)+n(1-p)p^{n-1}+n(1-p)p^{2n-2} }{(1-p^n)^2} \\
            &= \frac{n}{2} \cdot \frac{ p^{n-2} \pa{ -p+(n-1)(1-p)+p^2+np(1-p) } +p^{2n-2}-1 }{(1-p^n)^2} \\
            &= \frac{n}{2} \cdot \frac{g_1(p)}{(1-p^n)^2} < 0,
        \end{align*}
        so $g_0$ is strictly decreasing, and $g_0(p) > g_0(1) = 1$.
        Therefore, $n(1-p)(1+p^{n-1}) > 2(1-p^n)$. By \Cref{eq:f_derivative}, this implies that $f'(p) < 0$, which concludes the proof of \Cref{claim:f_decreasing}.
	\end{proof}
 
	\begin{claim} \label{claim:A_increasing}
		Function $A$ is (strictly) increasing in $\gamma$.
	\end{claim}
	\begin{proof}
		Using \cref{eq:new_expression_A}, we obtain
        \begin{equation*}
            \frac{dA(\gamma)}{d\gamma} = \frac{n}{(1-s)(1 + \gamma)^2} > 0,
        \end{equation*}
        from which \Cref{claim:A_increasing} follows.
	\end{proof}

    \begin{figure} [htbp]
        \centering
        \begin{subfigure}{.3\textwidth}
            \centering
            \includegraphics[width=\linewidth]{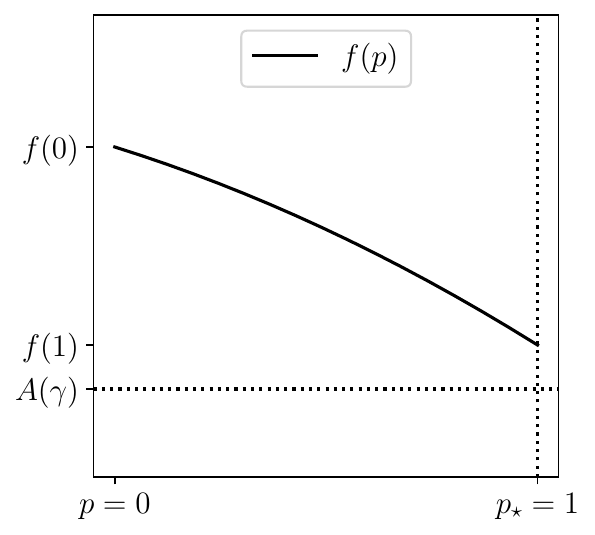}
            \caption{$\gamma = 0 < \gamma_1$.}
            \label{fig:function_f_gamma1}
        \end{subfigure}
        \hfill
        \begin{subfigure}{.3\textwidth}
            \centering
            \includegraphics[width=\linewidth]{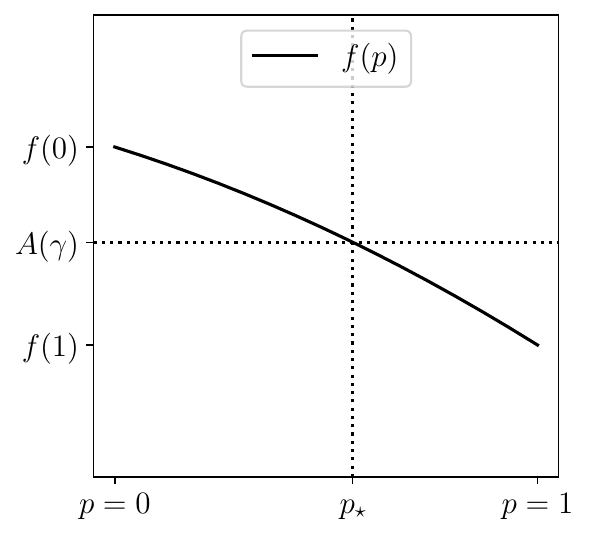}
            \caption{$\gamma = 1/2 \in [\gamma_1,\gamma_2]$.}
            \label{fig:function_f_gamma2}
        \end{subfigure}
        \hfill
        \begin{subfigure}{.3\textwidth}
            \centering
            \includegraphics[width=\linewidth]{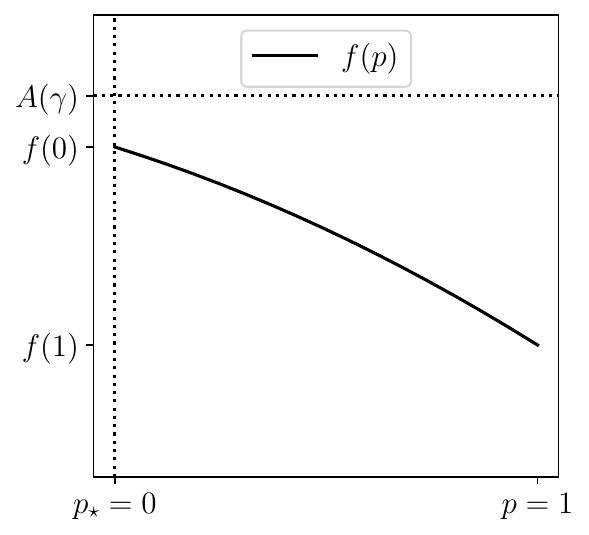}
            \caption{$\gamma = 2 > \gamma_2$.}
            \label{fig:function_f_gamma3}
        \end{subfigure}
        \caption{
        Visualization of Functions $f(p)$, $A(\gamma)$, and $p_\star(\gamma)$ illustrating their roles in the proof. The X-axis denotes the variable $p$, while each subfigure corresponds to a distinct value of $\gamma$. These figures were generated for $n=4$ and $s=2/5$.
        }
        \label{fig:function_f}
    \end{figure}
    See \Cref{fig:function_f} for an overview of the following arguments.
	
		\paragraph{Proof of (a).} Assume $n\geq 2$, and consider the case that $\gamma < \gamma_1$ (\Cref{fig:function_f_gamma1}). By Claims \ref{claim:simple_computations}, \ref{claim:f_decreasing} and \ref{claim:A_increasing}, and \Cref{claim:lim_f},
        for every~$p \in [0,1]$,
    	\begin{equation*}
    		f(p) \geq f(1) = -\frac{n(n-3)}{2} = A(\gamma_1) > A(\gamma).
    	\end{equation*}
    	By \Cref{claim:new_nash_condition}, this implies that $\payoff{1}{p} > \payoff{0}{p}$.
        Thus, for every $p<1$ and every~$q \in [0,1]$,
        \begin{equation} \label{eq:domination}
            \payoff{p}{q} = p \, \payoff{1}{q} + (1-p) \, \payoff{0}{q} < \payoff{1}{q}.
        \end{equation}
        On the one hand, \Cref{eq:domination} implies that for every~$p < 1$, $p$ cannot satisfy neither condition (i) nor (ii) in the definition of ESS.
        On the other hand, \Cref{eq:domination} implies that $p_\star = 1$ will always satisfy condition (i) in the definition of ESS.        
        Finally, we conclude that on $[0,\gamma_1)$, $p_\star(\gamma) = 1$ is the only ESS.

        \paragraph{} Next, consider the case that $\gamma > \gamma_2$ (\Cref{fig:function_f_gamma3}). By Claims \ref{claim:simple_computations}, \ref{claim:f_decreasing} and \ref{claim:A_increasing}, for every~$p \in [0,1]$,
    	\begin{equation*}
    		f(p) \leq f(0) = 1 = A(\gamma_2) < A(\gamma).
    	\end{equation*}
    	By \Cref{claim:new_nash_condition}, this implies that $\payoff{1}{p} < \payoff{0}{p}$.
        Similarly, we conclude that on $(\gamma_2,+\infty]$, $p_\star(\gamma) = 0$ is the only ESS.

        \paragraph{Proof of (b).} Consider the case that $n\geq 3$ and $\gamma_1 \leq \gamma \leq \gamma_2$ (\Cref{fig:function_f_gamma2}).
        By Claim~\ref{claim:f_decreasing}, $f : [0,1] \mapsto [f(1),f(0)]$ is a bijection, and we can consider the inverse function $f^{-1} : [f(1),f(0)] \mapsto [0,1]$.
        Moreover, by Claims \ref{claim:simple_computations}, \ref{claim:f_decreasing} and \ref{claim:A_increasing}, and \Cref{claim:lim_f},
        \begin{equation*}
		      f(1) = A(\gamma_1) \leq A(\gamma) \leq A(\gamma_2) = f(0).
	    \end{equation*}
        Therefore, there is a unique $p_\star \in [0,1]$ such that $f(p_\star) = A(\gamma)$.
        By \Cref{claim:f_decreasing,claim:new_nash_condition}, we have
        \begin{align*}
            f(p_\star) = A(\gamma) &\implies \payoff{1}{p_\star} = \payoff{0}{p_\star}, \\
            \text{for every } q < p_\star, f(q) > f(p_\star) &\implies \payoff{1}{p_\star} > \payoff{0}{p_\star}, \\
            \text{for every } q > p_\star, f(q) < f(p_\star) &\implies \payoff{1}{p_\star} < \payoff{0}{p_\star}.
        \end{align*}
        By \Cref{lem:ESS_sufficient}, this implies that $p_\star$ is the unique ESS.
        
        As a function of $\gamma$ on the interval $[\gamma_1,\gamma_2]$, $p_\star$ satisfies $p_\star(\gamma) = f^{-1}(A(\gamma))$.
        Function $f$ is continuously differentiable, and the derivative is non-zero by \Cref{claim:f_decreasing}, so $f^{-1}$ is continuously differentiable.
        Moreover, $A$ is also continuously differentiable.
        Therefore, $p_\star$ is continuously differentiable.
        Finally, $p_\star$ verifies $p_\star(\gamma_1) = f^{-1}(A(\gamma_1)) = f^{-1}(f(1)) = 1$, and $p_\star(\gamma_2) = f^{-1}(A(\gamma_2)) = f^{-1}(f(0)) = 0$.
\end{proof}

\begin{lemma} \label{lem:weak_braess}
    If $3 \leq n < \min \left\{ 1+\frac{1}{s}, 1+\frac{2}{1-s} \right\}$, then there exists $0 \leq \gamma_{\min} < \gamma_{\max}$ such that $\pi_\star(\gamma)$ is decreasing on the interval $[\gamma_{\min},\gamma_{\max}]$.
\end{lemma}

\begin{proof}
    Since $n \geq 3$, by definition, $\gamma_1 < \gamma_2$ and so $[\gamma_1 , \gamma_2]$ is a non-empty interval.
    We have that
    \begin{equation*}
        n \leq 1+\frac{2}{1-s} \iff \frac{2}{(n-1)(1-s)}-1 \geq 0 \iff \gamma_1 \geq 0.
    \end{equation*}
    Moreover,
    \begin{equation*}
        n < 1+\frac{1}{s} \iff \frac{2}{(n-1)(1-s)} > \frac{n}{(n-1)(1-s)}-1 \iff 1+\gamma_1 > \gamma_2.
    \end{equation*}
   Next, note that by definition: \begin{equation*}
        \pi_\star(\gamma) = p_\star(\gamma) \cdot \payoff{1}{p_\star(\gamma)}(\gamma) + (1-p_\star(\gamma)) \cdot \payoff{0}{p_\star(\gamma)}(\gamma),
    \end{equation*}
    
    By assumption on $n$, we know that $\gamma_1\geq 0$ and $1+\gamma_1>\gamma_2$.
    Since both $\payoff{1}{p}(\gamma)$ and $\payoff{0}{p}(\gamma)$ are continuously differentiable in $p$ and in $\gamma$ (from their expression in \Cref{lem:expectation_payoff}), and
        since  $p_\star(\gamma)$ is continuously differentiable in $\gamma$ on $[\gamma_1,\gamma_2]$ (by statement (b) in \Cref{lem:equilibrium_condition}),
    then $\pi_\star(\gamma)$
    is continuously differentiable in $\gamma$ on $[\gamma_1,\gamma_2]$. Moreover, it satisfies
    $\pi_\star(\gamma_1) = F_\producer = 1+\gamma_1$ (since $p_\star(\gamma_1) = 1$), and $\pi_\star(\gamma_2) = F_\scrounger =\gamma_2 < \pi_\star(\gamma_1)$ (since $p_\star(\gamma_2) = 0$).
    Therefore, we can find an interval $[\gamma_{\min},\gamma_{\max}] \subseteq [\gamma_1,\gamma_2]$ on which $\pi_\star(\gamma)$ is decreasing, which concludes the proof of \Cref{lem:weak_braess}. 
\end{proof}

\begin{proof} [Proof of \Cref{thm:foraging_braess}]
    When $n=3$,
    \begin{equation*}
        n < \min \left\{ 1+\frac{1}{s}, 1+\frac{2}{1-s} \right\} \iff 0 < s < \frac{1}{2},
    \end{equation*}
    and the first item in \Cref{thm:foraging_braess} follows as a special case of \Cref{lem:weak_braess}.

    When $n=2$, $\gamma_1 = \gamma_2 = \gamma_s = \frac{1+s}{1-s}$.
    By statement (a) in \Cref{lem:equilibrium_condition}, for every~$\gamma < \gamma_s$, there is a unique ESS satisfying $p_\star(\gamma) = 1$ and so $\pi_\star(\gamma) = 1+\gamma$.
    Similarly, for every~$\gamma > \gamma_s$, there is a unique ESS satisfying $p_\star(\gamma) = 0$ and so $\pi_\star(\gamma) = \gamma$.
    Therefore, $\pi_\star$ is increasing on $[0,\gamma_s)$ and on $(\gamma_s,+\infty)$.
    Moreover, let $\epsilon \in (0,1/2)$. Since $s \geq 0$, we have $\gamma_s \geq 1$, and
    \begin{equation*}
        \pi_\star(\gamma_s-\epsilon) = 1+\gamma_s-\epsilon > \gamma_s+\frac{1}{2} > \gamma_s+\epsilon = \pi_\star(\gamma_s+\epsilon), 
    \end{equation*}
    which establishes the second item in \Cref{thm:foraging_braess}, and thus concludes the proof of theorem.
\end{proof}

\subsection{Technical Claims}

\begin{claim} \label{claim:expectation_binomial_1}
    Let~$X \sim \calB(n,p)$. If $0 < p \leq 1$, then
    \begin{equation*}
        \bbE \pa{\frac{1}{1+X}} = \frac{1-(1-p)^{n+1}}{(n+1)p}.
    \end{equation*}
    Moreover, if $p = 0$, then $\bbE \pa{\frac{1}{1+X}} = 1$.
\end{claim}
\begin{proof}
    The claim holds trivially for~$p=0$. Consider the case that~$p>0$.
    \begin{align*}
        \bbE \pa{\frac{1}{1+X}} &= \sum_{k=0}^n \frac{1}{1+k} \Pr(X=k) \\
        &= \sum_{k=0}^n \frac{1}{1+k} \cdot \binom{n}{k} p^k (1-p)^{n-k} \\
        &= \frac{1}{(n+1)p} \sum_{k=0}^n \binom{n+1}{k+1} p^{k+1} (1-p)^{(n+1)-(k+1)} & \text{using } \binom{n}{k} = \binom{n+1}{k+1} \cdot \frac{k+1}{n+1}.
    \end{align*}
    By setting~$k' = k+1$, we can rewrite the sum
    \begin{equation*}
        \sum_{k=0}^n \binom{n+1}{k+1} p^{k+1} (1-p)^{(n+1)-(k+1)} = \sum_{k'=0}^{n+1} \binom{n+1}{k'} p^{k'} (1-p)^{(n+1)-k'} - (1-p)^{n+1} = 1-(1-p)^{n+1},
    \end{equation*}
    which concludes the proof of \Cref{claim:expectation_binomial_1}.
\end{proof}

\begin{claim} \label{claim:expectation_binomial_3}
    Let~$X \sim \calB(n,p)$. If $0 \leq p < 1$, then
    \begin{equation*}
        \bbE \pa{\frac{X}{2+n-X}}
        = p \cdot \frac{(n+1)(1-p)+p^{n+1}-1}{(n+1)(1-p)^2}.
    \end{equation*}
    Moreover, if $p = 1$, then $\bbE \pa{\frac{X}{2+n-X}} = \frac{n}{2}$.
\end{claim}
\begin{proof}
    The claim holds trivially for~$p=1$. Consider the case that~$p<1$.
    Let $q = 1-p > 0$. We have
    \begin{align*}
        &\bbE \pa{\frac{X}{2+n-X}} \\
        &= \sum_{k=0}^n ~ \frac{k}{n-k+2} \Pr(X=k) = \sum_{k=1}^n ~ \frac{k}{n-k+2} \cdot \binom{n}{k} p^k (1-p)^{n-k} \\
        &= \sum_{k=0}^{n-1} ~ \frac{n-k}{k+2} \cdot \binom{n}{k} q^k (1-q)^{n-k} & k \mapsto n-k,~ q =1-p\\
        &= \sum_{k=1}^n ~ \frac{n-k+1}{k+1} \cdot \binom{n}{k-1} q^{k-1} (1-q)^{n-k+1} & k \mapsto k-1\\
        &= \sum_{k=1}^n ~ \frac{k}{k+1} \cdot \binom{n}{k} q^{k-1} (1-q)^{n-k+1} & \text{using } \binom{n}{k-1} = \binom{n}{k} \cdot \frac{k}{n-k+1}\\
        &= \sum_{k=1}^n ~ \frac{k}{n+1} \cdot \binom{n+1}{k+1} q^{k-1} (1-q)^{n-k+1} & \text{using } \binom{n}{k} = \binom{n+1}{k+1} \cdot \frac{k+1}{n+1}\\
        &= \sum_{k=2}^{n+1} ~ \frac{k-1}{n+1} \cdot \binom{n+1}{k} q^{k-2} (1-q)^{n-k+2} & k \mapsto k-1\\
        &= \frac{1-q}{q^2 (n+1)} \cdot \sum_{k=1}^{n+1} (k-1) \binom{n+1}{k} q^{k} (1-q)^{(n+1)-k}. &
    \end{align*}
    By expectation of the binomial distribution,
    \begin{equation*}
        \sum_{k=1}^{n+1} k \binom{n+1}{k} q^{k} (1-q)^{(n+1)-k} = (n+1)q.
    \end{equation*}
    Moreover, by the binomial theorem,
    \begin{equation*}
        \sum_{k=1}^{n+1} \binom{n+1}{k} q^{k} (1-q)^{(n+1)-k} = 1-(1-q)^{n+1}.
    \end{equation*}
    Putting every equation together, we obtain
    \begin{equation*}
        \bbE \pa{\frac{X}{2+n-X}} = (1-q) \cdot \frac{(n+1)q+(1-q)^{n+1}-1}{(n+1)q^2},
    \end{equation*}
    which concludes the proof of \Cref{claim:expectation_binomial_3}.
\end{proof}

\section{Analysis of the Company Game}

\subsection{Preliminaries}
Following classical notations from game theory, we define, for $n=2$ players:
\begin{center}
\begin{tabular}{crl}
    {\bf (Reward)} & $ R(\gamma,s,c,p,a) ~~=$ & $\payoff{1}{1}$ \\
    {\bf (Sucker)} & $ S(\gamma,s,c,p,a) ~~=$ & $\payoff{1}{0}$ \\
    {\bf (Temptation)} & $ T(\gamma,s,c,p,a) ~~=$ & $\payoff{0}{1}$ \\
    {\bf (Punishment)} & $ P(\gamma,s,c,p,a) ~~=$ & $\payoff{0}{0}$ \\
\end{tabular}
\end{center}
For simplicity, we do not mention~$(\gamma,s,c,p,a)$ when there is no risk of confusion.

The following result is well-known in game theory folklore. However, we provide a proof here for the sake of completeness. 
\begin{theorem} [Game of Chicken] \label{thm:chichen_game_equilibrium_payoff}
    If $n=2$ and~$T>R>S>P$, then there is a unique ESS, that satisfies
    \begin{equation} \label{eq:chicken_nash_payoff}
        \pi_\star = \frac{ST-RP}{S+T-R-P}.
    \end{equation}
\end{theorem}
\begin{proof}
    We have, by definition $\payoff{1}{p} = p \, R + (1-p) \, S$ and $\payoff{0}{p} = p \, T + (1-p) \, P$. Note that
    \begin{equation} \label{eq:p_star_chicken}
        \payoff{1}{p} = \payoff{0}{p} \iff p = \frac{S-P}{S+T-R-P}.
    \end{equation}
    Define
    \begin{equation*}
        p_\star = \frac{S-P}{S+T-R-P} = \frac{1}{1+\frac{T-R}{S-P}}
    \end{equation*}
    with $T-R > 0$, $S - P > 0$, so $p_\star \in [0,1]$.
    Let
    \begin{equation}\label{eq:pi}
        \pi_\star = \payoff{1}{p_\star} = \payoff{0}{p_\star}.
    \end{equation}
    By assumption in the theorem, \[\frac{d}{dp} \payoff{1}{p} = R-S < T-P = \frac{d}{dp} \payoff{0}{p},\] and hence, by Eq.~\eqref{eq:pi}, 
    for every~$q < p_\star$, $\payoff{1}{q} > \payoff{0}{q}$,  and for every~$q > p_\star$, $\payoff{1}{q} < \payoff{0}{q}$.
    By \Cref{lem:ESS_sufficient}, this, together with Eq.~\eqref{eq:pi}, implies that $p_\star$ is a unique ESS.
    To conclude the proof of \Cref{thm:chichen_game_equilibrium_payoff}, we just check that $\pi_\star$ satisfies \Cref{eq:chicken_nash_payoff}.
\end{proof}

\subsection{There is no Reverse-Correlation phenomenon for $\phi : x \mapsto x$} \label{observation}

\begin{observation} \label{thm:no_braess_example}
    Consider the case that there are $n=2$ players and that $\phi : x \mapsto x$.
    Let $c\geq 0$, $s\in[1/2,1]$, $p,a\in [0,1]$, and set~$\gamma_0 = c/(p s (1-a))$.
    Then for every $\gamma \neq \gamma_0$, there is a unique ESS.
    Moreover, the payoff $\pi_\star$ and the total production $\Gamma_\star$ corresponding to the ESS are both strictly increasing functions of~$\gamma$.
\end{observation}
\begin{proof}
Recall that in the Company game,
\begin{equation} \label{eq:general_payoff}
    \pi_1 = \phi(s q_1 + (1-s) q_2) - c_1.
\end{equation}


\begin{figure}[htbp]
    \centering
    \includegraphics[width=0.35\linewidth]{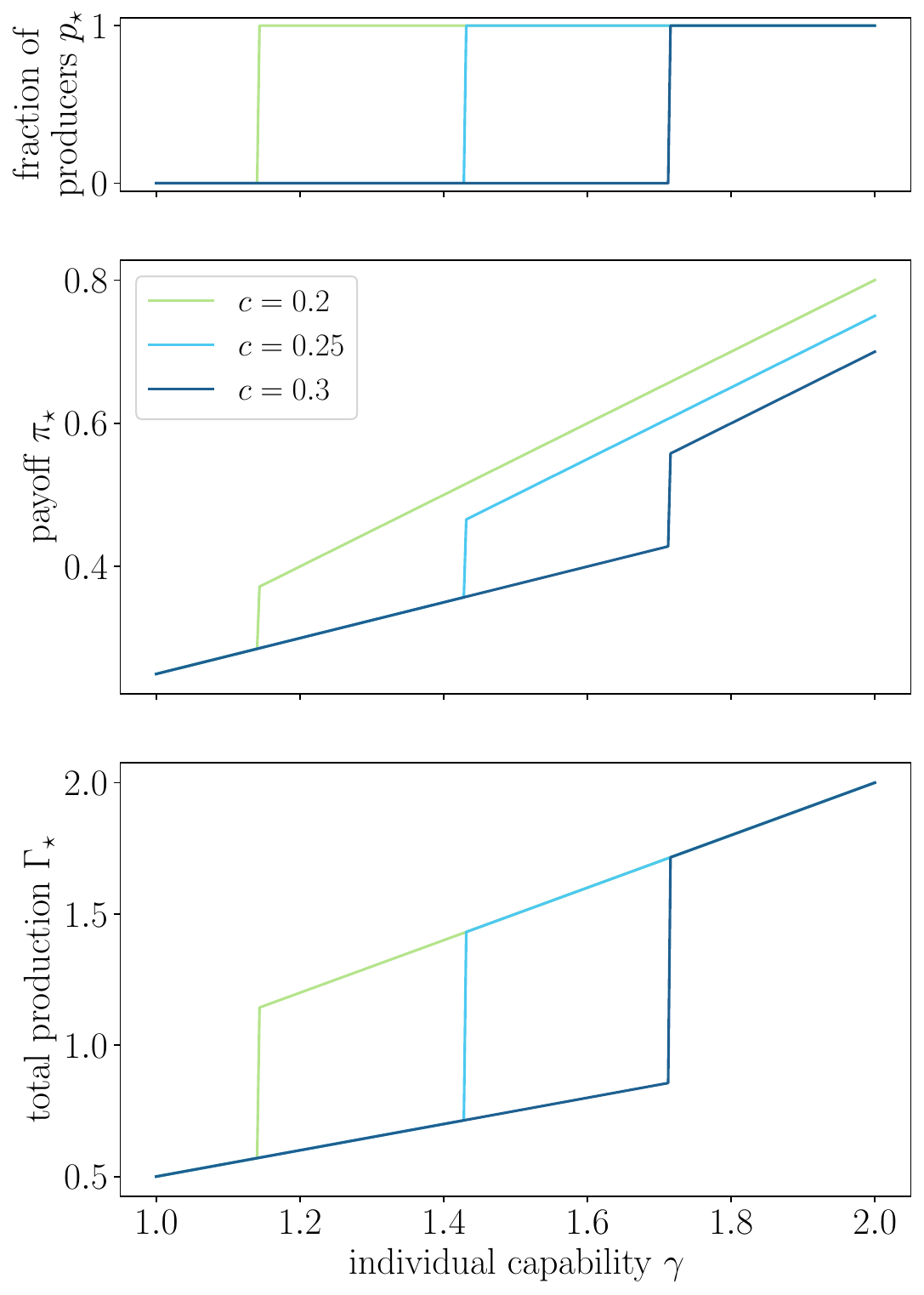}
    \caption{{\bf The Company game with linear $\phi$.}  
    Simulating the Company game with $\phi: x \mapsto x$,
    $n=2, a=p=\frac{1}{2}, s=0.7$, and $c=0.25$. The simulation reveals no Reverse-Correlation phenomenon.}
    \label{fig:company_no_braess}
\end{figure}

Consider the case that~$\phi : x \mapsto x$. For every $s\in [1/2,1]$,  
\Cref{eq:general_payoff} gives
\begin{equation*}
    \bbE(\pi_1) = s \bbE(q_1) + (1-s) \bbE(q_2) - c_1.
\end{equation*}
Therefore, 
\begin{align}
    R(\gamma,s,c,p,a) &= s \cdot(p \gamma) + (1-s) \cdot(p \gamma) - c = p \gamma - c, \label{eq:R0} \\
    S(\gamma,s,c,p,a) &= s \cdot(p \gamma) + (1-s) \cdot(p a \gamma) - c = p \gamma (s+a-s a) - c, \label{eq:S0} \\
    T(\gamma,s,c,p,a) &= s \cdot(p a \gamma) + (1-s) \cdot(p \gamma) = p \gamma (1-s+s a), \label{eq:T0}\\
    P(\gamma,s,c,p,a) &= s \cdot(p a \gamma) + (1-s) \cdot(p a \gamma) = p a \gamma. \label{eq:P0}
\end{align}
Recall that~$\gamma_0 = c/(p s (1-a))$.
\begin{itemize}
    \item If $\gamma < \gamma_0$, then $T>R$ and $P>S$, in which scrounger is a dominant strategy. Therefore, there is a unique ESS, and we have:
    \[\pi_{\star}=\Gamma_\star=P = pa\gamma.\] 
    In particular, these values are increasing in  $\gamma$. 
    \item If $\gamma > \gamma_0$, then $R>T$ and $S>P$, and hence producer is a dominant strategy. Therefore, there is a unique ESS, and 
     \[\pi_{\star}=R = p \gamma - c, \mbox{~~and~~}\Gamma_{\star}=p\gamma.\]
    In particular, both these values are increasing in $\gamma$.
    \item If $\gamma = \gamma_0$, then $R=T$ and $P=S$, which implies that no player can unilaterally change its payoff.
    Indeed, for every~$p,q \in [0,1]$,
    \begin{equation*}
        \payoff{p}{q} = pqR + (1-p)qT + p(1-q)S + (1-p)(1-q)P = qR + (1-q)S,
    \end{equation*}
    so for every~$p,p',q \in [0,1]$, $\payoff{p}{q} = \payoff{p'}{q}$.
    In this degenerate case, neither condition (i) nor (ii) in the definition of ESS can be satisfied, so there is no ESS. 
\end{itemize}
To conclude the proof of \Cref{thm:no_braess_example}, we only need to show that $\pi_\star$ and $\Gamma_\star$ do not decrease at the discontinuity point $\gamma = \gamma_0$.
Since $\gamma_0 \geq c/(p(1-a))$, we have 
\begin{equation*}
    \lim_{\epsilon \rightarrow 0^+} \Gamma_\star(\gamma_0-\epsilon) = \lim_{\epsilon \rightarrow 0^+} \pi_\star(\gamma_0-\epsilon) = ap\gamma_0 \leq p\gamma_0-c = \lim_{\epsilon \rightarrow 0^+} \pi_\star(\gamma_0+\epsilon) \leq \lim_{\epsilon \rightarrow 0+} \Gamma_\star(\gamma_0+\epsilon).
\end{equation*}
Thus, overall, the payoffs of players at equilibrium, $\pi_{\star}$, and the total production, $\Gamma_{\star}$, are both increasing in~$\gamma$.
\end{proof}

\subsection{The Reverse-Correlation phenomenon in the Company game}

In this section, we demonstrate the Reverse-Correlation phenomenon in the Company game for two utility functions. Specifically, we first prove that the Reverse-Correlation phenomenon can occur when assuming 
the utility function $\phi : x \mapsto 1-\exp(-2 \, x)$.
Then, in  \Cref{sec:alternative_company_RC} we provide simulations that demonstrate the Reverse-Correlation phenomenon assuming the utility function $\phi : x \mapsto \min(1,x)$.

\begin{theorem} \label{thm:company_braess}
    Consider the Company game with $n=2$, $\phi : x \mapsto 1-\exp(-2 \, x)$, $p=1/2$, $a=1/2$.
    For every~$s < 1$, there exist $c_0>0$, and $\gamma_{\min} , \gamma_{\max} > 1$, for which there is a unique ESS such that $\pi_\star$ is decreasing in $\gamma$ on the interval~$[\gamma_{\min},\gamma_{\max}]$.
\end{theorem}
\begin{proof}
Let
\begin{equation*}
    \gamma_i = \begin{cases}
        \gamma & \if \text{Player $i$ is a producer}, \\
        \gamma/2 & \text{otherwise}.
    \end{cases}
\end{equation*}
By definition, if player $i$ succeeds in producing a product (which happens with probability $p=1/2$) then the quality of its product is $\gamma_i$. 
Hence, 
\Cref{eq:general_payoff} gives 
\begin{equation*}
    \bbE(\pi_1) = \frac{1}{4} \pa{\phi(s \gamma_1 + (1-s) \gamma_2) + \phi(s \gamma_1) + \phi((1-s) \gamma_2)} - c_1.
\end{equation*}
 Plugging in $\phi(x) = 1-\exp(-2x)$, we obtain
\begin{align}
    R(\gamma,s,c) &= \frac{1}{4} \pa{ 3 - e^{-2\gamma} - e^{-2s \gamma} - e^{-2(1-s)\gamma} } - c, \label{eq:R} \\
    S(\gamma,s,c) &= \frac{1}{4} \pa{ 3 - e^{-(1+s)\gamma} - e^{-2s \gamma} - e^{-(1-s)\gamma} } - c, \label{eq:S} \\
    T(\gamma,s,c) &= \frac{1}{4} \pa{ 3 - e^{-(2-s)\gamma} - e^{-s \gamma} - e^{-2(1-s)\gamma} }, \label{eq:T}\\
    P(\gamma,s,c) &= \frac{1}{4} \pa{ 3 - e^{-\gamma} - e^{-s \gamma} - e^{-(1-s)\gamma} }. \label{eq:P}
\end{align}
Note that $R, S, T$ and $P$ are all increasing functions of $\gamma$.
Consequently, if the strategies of Players 1 and 2 remain unchanged, then $\bbE(\pi_1)$ is also increasing in $\gamma$. However, we will show that at equilibrium, the tendency of the player to be a scrounger increases in $\gamma$ to such an extent that ultimately reduces $\bbE(\pi_1)$. 

The next step towards proving \Cref{thm:company_braess} is to show that for some specific values of~$s$ and $c$, the Company game is in fact a game of chicken.
\begin{lemma} \label{lem:company_is_chicken}
    For every~$s < 1$, there exists a value $c_0 = c_0(s)$ and an interval~$[\gamma_{\min},\gamma_{\max}]$ such that for every~$\gamma \in [\gamma_{\min},\gamma_{\max}]$, 
    \begin{equation} \label{eq:chicken_structure}
        T(\gamma,s,c_0) > R(\gamma,s,c_0) > S(\gamma,s,c_0) > P(\gamma,s,c_0).
    \end{equation}
    In particular, by \Cref{thm:chichen_game_equilibrium_payoff}, this implies that for every~$\gamma \in [\gamma_{\min},\gamma_{\max}]$, there is a unique ESS satisfying
    \begin{equation*}
        \pi_\star(\gamma,s,c_0) = \frac{ST-RP}{S+T-R-P}.
    \end{equation*}
\end{lemma}

Before proving \Cref{lem:company_is_chicken}, we need two preliminary technical results.
The next claim implies, in particular, that if $c=0$ then a producer is a dominant strategy.
\begin{claim} \label{claim:producers_are_better_friends}
    For all~$\gamma,s$ such that $s < 1$, $R(\gamma,s,0) > S(\gamma,s,0),T(\gamma,s,0)$, and $S(\gamma,s,0),T(\gamma,s,0) > P(\gamma,s,0)$.
\end{claim}
\begin{proof}
    By pairwise comparison of the terms in \Cref{eq:R,eq:S,eq:T,eq:P}.
\end{proof}

The next claim implies that $T-P > R-S$, or in other words, that scroungers lose more than producers when the other player switches from producer to scrounger.
\begin{claim} \label{claim:towards_chicken}
    For all~$\gamma,s,c$ such that $s < 1$,
    \begin{equation*}
        S(\gamma,s,c)-P(\gamma,s,c) > R(\gamma,s,c)-T(\gamma,s,c).
    \end{equation*}
\end{claim}
\begin{proof}
    We have
    \begin{equation*}
        4(R-S) = \pa{ e^{-(1+s)\gamma} - e^{-2\gamma} } + \pa{ e^{-(1-s)\gamma} - e^{-2(1-s)\gamma} },
    \end{equation*}
    and
    \begin{equation*}
        4(T-P) = \pa{ e^{-\gamma} - e^{-(2-s)\gamma} } + \pa{ e^{-(1-s)\gamma} - e^{-2(1-s)\gamma} }.
    \end{equation*}
    Factoring by $e^{-\gamma}$, this gives
    \begin{equation} \label{eq:save_for_later}
    	(T-P) - (R-S) = \frac{e^{-\gamma}}{4} \pa{ 1 + e^{-\gamma} - e^{-s\gamma} - e^{-(1-s)\gamma} }.
    \end{equation}
    Factoring again by $e^{-\gamma}$, and using the convexity of the function  $e^x+e^{\gamma-x}$, we obtain
    \begin{equation*}
    	(T-P) - (R-S) = \frac{e^{-2\gamma}}{4} \pa{ e^{\gamma} + 1 - \pa{ e^{s \gamma} + e^{(1-s) \gamma} } } > 0,
    \end{equation*}
    which concludes the proof of \Cref{claim:towards_chicken}.
\end{proof}

\begin{proof} [Proof of \Cref{lem:company_is_chicken}]
    Let us fix~$\gamma_0 > \max(1,\ln(1+\sqrt{2}) / s)$. By \Cref{claim:producers_are_better_friends,claim:towards_chicken},  we can take $c_0 = c_0(s)$ such that
    \begin{equation*}
        0 < R(\gamma_0,s,0)-T(\gamma_0,s,0) < c_0 < S(\gamma_0,s,0)-P(\gamma_0,s,0).
    \end{equation*}
    As a consequence,
    \begin{equation*}
        R(\gamma_0,s,c_0) = R(\gamma_0,s,0) - c_0 < T(\gamma_0,s,0)=T(\gamma_0,s,c_0),
    \end{equation*}
    and
    \begin{equation*}
        S(\gamma_0,s,c_0) = S(\gamma_0,s,0) - c_0 > P(\gamma_0,s,0)=P(\gamma_0,s,c_0).
    \end{equation*}
    Finally, by \Cref{claim:producers_are_better_friends}, we have that
    \begin{equation*}
        T(\gamma_0,s,c_0) > R(\gamma_0,s,c_0) > S(\gamma_0,s,c_0) > P(\gamma_0,s,c_0).
    \end{equation*}
    By continuity, there exist $\gamma_{\min},\gamma_{\max}$ such that $\max(1,\ln(1+\sqrt{2}) / s) < \gamma_{\min} < \gamma_0 < \gamma_{\max}$ and for every~$\gamma \in [\gamma_{\min},\gamma_{\max}]$, \Cref{eq:chicken_structure} holds,
    which concludes the proof of \Cref{lem:company_is_chicken}.
\end{proof}

\begin{claim} \label{claim:obscure_computations}
   For every~$\gamma \in [\gamma_{\min},\gamma_{\max}]$,
   \begin{equation*}
       \pi_\star(\gamma,s,c_0) = 1 - c_0 \, e^{s \gamma} \cdot \pa{ \frac{e^{s \gamma}+1}{e^{s \gamma}-1}} = 1 - c_0 \, e^{s \gamma} \coth\pa{\frac{s \gamma}{2}}.
   \end{equation*}
\end{claim}
\begin{proof}
    We start from the expression of~$\pi_\star(\gamma,s,c_0)$ given by \Cref{lem:company_is_chicken}.
    First, we compute $ST-RP$ using \Cref{eq:R,eq:S,eq:T,eq:P}. For that purpose, we expand $ST$ and $RP$ separately, and then simplify.
    We have (each line corresponds to one term of $S$ multiplied by all the terms of $T$):
\begin{align*}
	16 \cdot ST &= \pa{3 - e^{-(1+s)\gamma} - e^{-2s \gamma} - e^{-(1-s)\gamma} - 4c_0} \cdot \pa{3 - e^{-(2-s)\gamma} - e^{-s \gamma} - e^{-2(1-s)\gamma}} \\
	&= 9 - 3e^{-(2-s)\gamma} - 3e^{-s\gamma} - 3e^{-2(1-s)\gamma} \\
	&- 3e^{-(1+s)\gamma} + e^{-3\gamma} + e^{-(1+2s)\gamma} + e^{-(3-s)\gamma} \\
	&- 3e^{-2s\gamma} + e^{-(2+s)\gamma} + e^{-3s\gamma} + e^{-2\gamma} \\
	&- 3e^{-(1-s)\gamma} + e^{-(3-2s)\gamma} + e^{-\gamma} + e^{-3(1-s)\gamma} \\
	&-12c_0 + 4c_0e^{-(2-s)\gamma} + 4c_0e^{-s\gamma} + 4c_0 e^{-2(1-s)\gamma}.
\end{align*}
Similarly,
\begin{align*}
	16 \cdot RP &= \pa{ 3 - e^{-2\gamma} - e^{-2s \gamma} - e^{-2(1-s)\gamma} - 4c_0} \cdot \pa{ 3 - e^{-\gamma} - e^{-s \gamma} - e^{-(1-s)\gamma} } \\
	&= 9 - 3e^{-\gamma} - 3e^{-s\gamma} - 3e^{-(1-s)\gamma} \\
	&- 3e^{-2\gamma} + e^{-3\gamma} + e^{-(2+s)\gamma} + e^{-(3-s)\gamma} \\
	&- 3e^{-2s\gamma} + e^{-(1+2s)\gamma} + e^{-3s\gamma} + e^{-(1+s)\gamma} \\
	&- 3e^{-2(1-s)\gamma} + e^{-(3-2s)\gamma} + e^{-(2-s)\gamma} + e^{-3(1-s)\gamma} \\
	&-12c_0 + 4c_0e^{-\gamma} + 4c_0e^{-s\gamma} + 4c_0 e^{-(1-s)\gamma}.
\end{align*}
When computing the difference, many terms disappear, leaving us with:
\begin{equation*}
	16 \cdot (ST-RP) = 4 \pa{ e^{-\gamma} + e^{-2\gamma} - e^{-(1+s)\gamma} - e^{-(2-s)\gamma} - c_0 \pa{ e^{-\gamma} + e^{-(1-s)\gamma} - e^{-(2-s)\gamma} - e^{-2(1-s)\gamma} } }.
\end{equation*}
Factoring the right hand side by $e^{-\gamma}$, we obtain
\begin{equation*}
	ST-RP = \frac{e^{-\gamma}}{4} \pa{ 1 + e^{-\gamma} - e^{-s\gamma} - e^{-(1-s)\gamma} - c_0 \pa{ 1 + e^{s\gamma} - e^{-(1-s)\gamma} - e^{-(1-2s)\gamma} } }.
\end{equation*}
Using \Cref{eq:save_for_later}, we obtain
\begin{equation}
	\frac{ST-RP}{S+T-R-P} = 1-c_0 \cdot \frac{1 + e^{s\gamma} - e^{-(1-s)\gamma} - e^{-(1-2s)\gamma}}{1 + e^{-\gamma} - e^{-s\gamma} - e^{-(1-s)\gamma} }.
\end{equation}
Factoring the numerator of the fraction by $e^{s\gamma}$ and rearranging, we get
\begin{equation*}
	1-c_0 e^{s\gamma} \cdot \frac{1 - e^{-(1-s)\gamma} + \pa{ e^{-s\gamma} - e^{-\gamma} }}{1 - e^{-(1-s)\gamma} - \pa{ e^{-s\gamma} - e^{-\gamma} }}.
\end{equation*}
Dividing both the numerator and denominator by $\pa{ e^{-s\gamma} - e^{-\gamma} }$, and using the fact that
\begin{equation*}
	\frac{1-e^{-(1-s)\gamma}}{e^{-s\gamma} - e^{-\gamma}} = e^{s \gamma} \cdot \pa{ \frac{e^{-s\gamma} - e^{-\gamma}}{e^{-s\gamma} - e^{-\gamma}} } = e^{s \gamma},
\end{equation*}
we finally get
\begin{equation*}
	\frac{ST-RP}{S+T-R-P} = 1-c_0 e^{s\gamma} \cdot \frac{e^{s \gamma}+1}{e^{s \gamma}-1},
\end{equation*}
which concludes the proof of \Cref{claim:obscure_computations}.
\end{proof}

\begin{claim} \label{claim:equilibrium_payoff_derivative}
   For every~$\gamma \in [\gamma_{\min},\gamma_{\max}]$,
   \begin{equation*}
       \frac{\partial}{\partial \gamma} \pi_\star(\gamma,s,c_0) = \frac{c_0 s \, e^{s \gamma}}{2} \cdot \frac{ 1 - \sinh(s \gamma)}{ \sinh \pa{\frac{s \gamma}{2}}^2 }.
   \end{equation*}
\end{claim}
\begin{proof}
We start from the expression of \Cref{claim:obscure_computations}, and derive using the fact that
\begin{equation*}
    \frac{d}{dx} \coth(x) = \frac{-1}{\sinh(x)^2}.
\end{equation*}
More precisely, by \Cref{claim:obscure_computations},
\begin{align*}
    \frac{\partial}{\partial \gamma} \pi_\star(\gamma,s,c_0) &= \frac{\partial}{\partial \gamma} \pa{1 - c_0 \, e^{s \gamma} \coth\pa{\frac{s \gamma}{2}}} \\
    &= -c_0 s e^{s \gamma} \coth\pa{\frac{s \gamma}{2}} + \frac{c_0 s e^{s \gamma}}{2 \sinh\pa{\frac{s \gamma}{2}}^2} \\
    &= \frac{c_0 s \, e^{s \gamma}}{2 \sinh\pa{\frac{s \gamma}{2}}^2} \cdot \pa{1 - 2 \coth \pa{\frac{s \gamma}{2}} \sinh \pa{\frac{s \gamma}{2}}^2}.
\end{align*}
Then, we observe that for every~$x \in \bbR$,
\begin{equation*}
    2 \coth(x) \sinh(x)^2 = 2 \pa{\frac{e^x+e^{-x}}{e^x-e^{-x}}} \pa{\frac{e^x-e^{-x}}{2}}^2 = \frac{(e^x+e^{-x})(e^x-e^{-x})}{2} = \frac{e^{2x}-e^{-2x}}{2} = \sinh(2x).
\end{equation*}
Plugging this in the last equation concludes the proof of \Cref{claim:equilibrium_payoff_derivative}.
\end{proof}
By definition, 
\begin{equation*}
    \gamma > \gamma_{\min} > \frac{\ln(1+\sqrt{2})}{s} = \frac{\sinh^{-1}(1)}{s},
\end{equation*}
so $\sinh(s \gamma) > 1$. By \Cref{claim:equilibrium_payoff_derivative}, this implies that $\frac{\partial}{\partial \gamma} \pi_\star(\gamma,s,c_0) < 0$ on the interval $[\gamma_{\min},\gamma_{\max}]$, which concludes the proof of \Cref{thm:company_braess}. 
\end{proof}




\begin{figure} [htbp]
    \centering   
    \begin{subfigure}{.45\textwidth}
        \centering
        \begin{subfigure}{\textwidth}
            \centering
            \includegraphics[width=\linewidth]{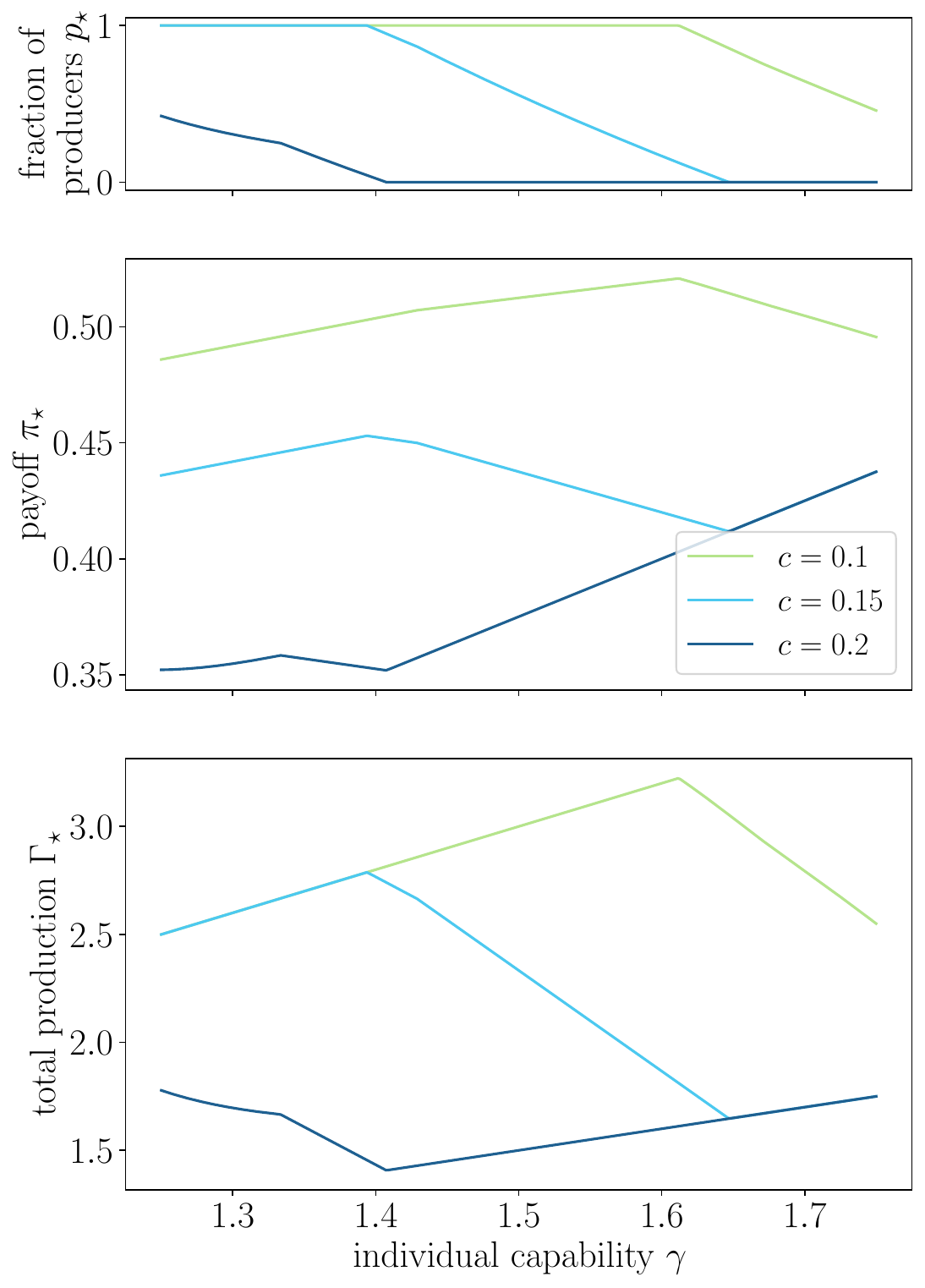}
            \caption{Payoffs and Total Production. $c=0.15$}
        \end{subfigure}
    \end{subfigure}
    \hfill
    \begin{subfigure}{.45\textwidth}
        \centering
        \begin{subfigure}{\textwidth}
            \centering
            \includegraphics[width=\linewidth]{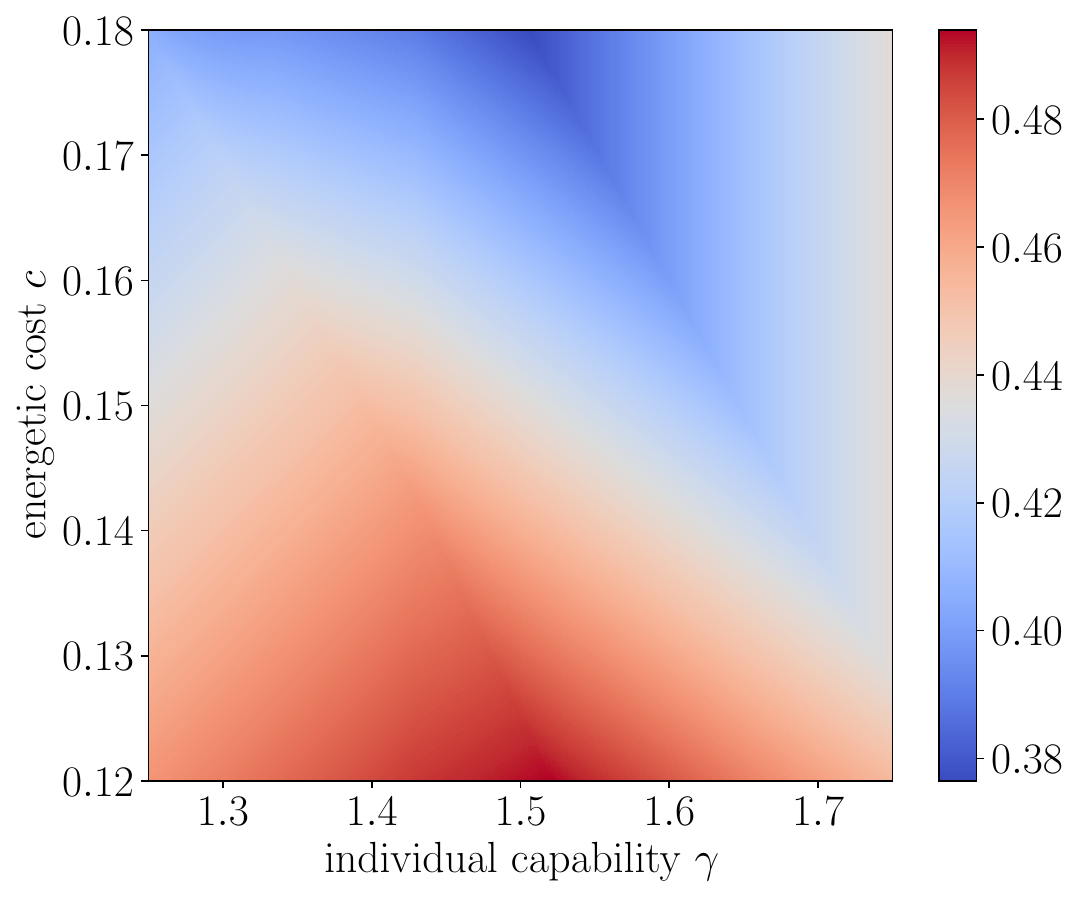}
            \caption{Payoffs.}
        \end{subfigure}
        \vfill
        \begin{subfigure}{\textwidth}
            \centering
            \includegraphics[width=\linewidth]{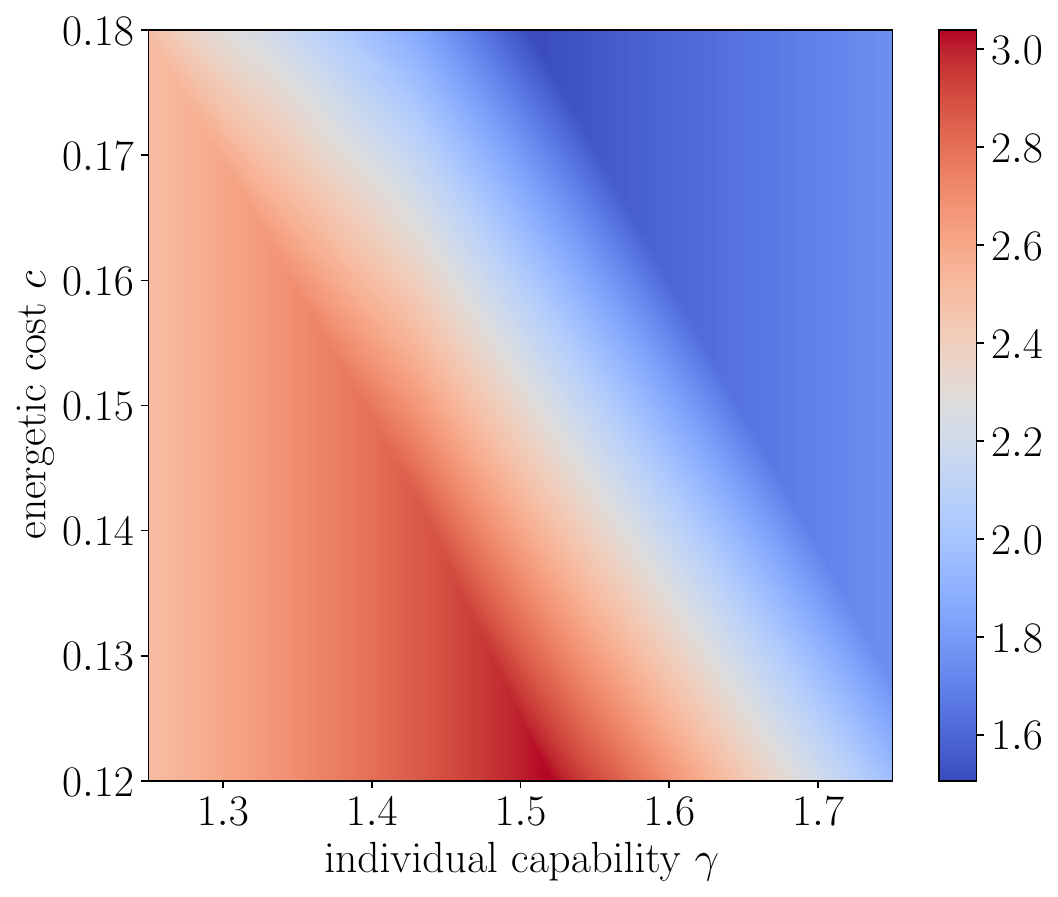}
            \caption{Total Production.}
        \end{subfigure}
    \end{subfigure}
    
    \caption{{\bf The Company game with the utility function $\phi: x \mapsto \min(1,x)$.} 
    Considering a scenario with $n=4$ players,  assuming $a=p=\frac{1}{2}$, and $s=0.7$.
    {\bf (a)} The graph presents the payoff $\pi_\star$ and the total production $\Gamma_\star(\gamma)$ as well as the probability of being a producer $p_\star(\gamma)$ at equilibria, as a function of individual capabilities $\gamma$ and for several values of $c$. The declines in the corresponding plots (of the payoff and the total production) display the Reverse-Correlation phenomenon.
    {\bf (b)} and {\bf (c)} The relationship between {\bf (b)} the payoff $\pi_\star$,  and {\bf (c)} the total production $\Gamma_\star$,  as a function (color scale) of individual capabilities $\gamma$ and the cost $c$ for production.}
    \label{sec:alternative_company_RC} 
\end{figure}

\clearpage

\section{A Necessary Condition for the Reverse-Correlation phenomenon}\label{sec:necessary}

We assume that the payoffs are positively correlated with the number of producers in the group, that is,
\begin{equation} \label{eq:producing_is_good}
    \text{for every~$q \in [0,1]$ and every~$\gamma \geq 0$, $p \mapsto \payoff{q}{p}(\gamma)$ is non-decreasing in $p$.}
\end{equation}
In addition, we assume that the payoffs are positively correlated with the parameter~$\gamma$, that is,
\begin{equation} \label{eq:gamma_is_good}
    \text{for every~$q \in [0,1]$ and every~$p \in [0,1]$, $\gamma \mapsto \payoff{q}{p}(\gamma)$ is non-decreasing in $\gamma$.}
\end{equation}
Under these assumptions, we identify the following necessary condition for the emergence of a Reverse-Correlation phenomenon.

\begin{theorem}\label{thm:necessary}
    For any PS model in which the payoff of producers does not depend on the strategies of other players, there is no Reverse-Correlation phenomenon.
    More precisely, if there are two values $\gamma_1,\gamma_2$ such that   $\gamma_1 < \gamma_2$ and two ESS 
 denoted $p_\star(\gamma_1)$ and $p_\star(\gamma_2)$, then the corresponding payoffs satisfy $\pi_\star(\gamma_1) \leq \pi_\star(\gamma_2)$.
\end{theorem}
\begin{proof}
    Fix a PS model.
    By assumption in \Cref{thm:necessary},
    \begin{equation} \label{eq:producer_insensitive}
        \text{For every~$\gamma \geq 0$, $p \mapsto \payoff{1}{p}(\gamma)$ does not depend on~$p$.}
    \end{equation}
    In what follows, we will simply write $\payoff{1}{p}(\gamma) = \pi_\producer(\gamma)$.
    By definition of ESS, and by \Cref{eq:producer_insensitive}, we have for every~$i \in \{1,2\}$:
    \addtocounter{equation}{1}
    \begin{align}
        p_\star(\gamma_i) = 0 &\implies \pi_\star(\gamma_i) = \payoff{0}{0}(\gamma_i) \geq \pi_\producer(\gamma_i), \label{eq:nash_definition1} \tag{\theequation.a} \\
        p_\star(\gamma_i) = 1 &\implies \pi_\star(\gamma_i) = \pi_\producer(\gamma_i) \geq \payoff{0}{1}(\gamma_i), \label{eq:nash_definition2} \tag{\theequation.b} \\
        p_\star(\gamma_i) \notin \{0,1\} &\implies \pi_\star(\gamma_i) = \pi_\producer(\gamma_i) =  \payoff{0}{p_\star(\gamma_i)}(\gamma_i), \label{eq:nash_definition3} \tag{\theequation.c}
    \end{align}
    where \Cref{eq:nash_definition1} holds because $\payoff{0}{0}(\gamma_i) \geq \payoff{1}{0}(\gamma_i)=\pi_\producer(\gamma_i)$.
    As a consequence of \Cref{eq:nash_definition2,eq:nash_definition3}, we have
    \begin{equation} \label{eq:nash_definition4}
        p_\star(\gamma_i) \neq 0 \implies \pi_\star(\gamma_i) = \pi_\producer(\gamma_i) \geq \payoff{0}{p_\star(\gamma_i)}(\gamma_i).
    \end{equation}
    Now, let us
    show that $\pi_\star(\gamma_1) \leq \pi_\star(\gamma_2)$.
    \begin{itemize}
        \item If $p_\star(\gamma_1) = p_\star(\gamma_2) = 0$, then
        \begin{equation*}
            \pi_\star(\gamma_1) \underset{\eqref{eq:nash_definition1}}{=} \payoff{0}{0}(\gamma_1) \underset{\eqref{eq:gamma_is_good}}{\leq} \payoff{0}{0}(\gamma_2) \underset{\eqref{eq:nash_definition1}}{=} \pi_\star(\gamma_2).
        \end{equation*}
        
        \item If $p_\star(\gamma_1) \neq 0$ and $p_\star(\gamma_2) \neq 0$, then
        \begin{equation*}
            \pi_\star(\gamma_1) \underset{\eqref{eq:nash_definition4}}{=} \pi_\producer(\gamma_1) \underset{\eqref{eq:gamma_is_good}}{\leq} \pi_\producer(\gamma_2) \underset{\eqref{eq:nash_definition4}}{=} \pi_\star(\gamma_2).
        \end{equation*}

        \item If $p_\star(\gamma_1) \neq 0$ and $p_\star(\gamma_2) = 0$, then
        \begin{equation*}
            \pi_\star(\gamma_1) \underset{\eqref{eq:nash_definition4}}{=} \pi_\producer(\gamma_1) \underset{\eqref{eq:gamma_is_good}}{\leq} \pi_\producer(\gamma_2) \underset{\eqref{eq:nash_definition1}}{\leq} \pi_\star(\gamma_2).
        \end{equation*}
        \item If $p_\star(\gamma_1) = 0$ and $p_\star(\gamma_2) \neq 0$, then
        \begin{equation*}
            \pi_\star(\gamma_1) \underset{\eqref{eq:nash_definition1}}{=} \payoff{0}{0}(\gamma_1) \underset{\eqref{eq:producing_is_good}}{\leq} \payoff{0}{p_\star(\gamma_2)}(\gamma_1) \underset{\eqref{eq:gamma_is_good}}{\leq} \payoff{0}{p_\star(\gamma_2)}(\gamma_2) \underset{\eqref{eq:nash_definition4}}{\leq} \pi_\star(\gamma_2).
        \end{equation*}
    \end{itemize}
    This concludes the proof of \Cref{thm:necessary}.
\end{proof}

\end{document}